\documentclass[preprint,10pt,numbers]{elsarticle}
\usepackage{color}
\newcommand{\mc}{\mathcal}            % Graph font
\newcommand{\G}{\mathcal{G}}   % Digraph symbol
\newcommand{\W}{\mathrm{W}}  % Weight function
\newcommand{\ew}{\mathsf{w}}   % Edge weight
\newcommand{\ot}{\leftarrow}    % Walk, right to left
\newcommand{\e}{\varepsilon}           % Operator lattice
\newcommand{\matr}[1]{\mathsf{#1}} % Matrix font
\newcommand{\I}{\matr{I}}                  % Identity matrix
\renewcommand{\OE}{\mathrm{OE}} % Ordered Exponential
\renewcommand{\H}{\matr{H}}   	 % Hamiltonian
\newcommand{\C}{\matr{C}}   		 % Time matrix
\newcommand{\M}{\matr{M}}   		 % Time matrix
\newcommand{\A}{\matr{A}}   		 % Motion generator
\newcommand{\T}{\matr{T}}   		 % Time translation
\renewcommand{\v}{\matr{v}}      	 % Complete vector
\newcommand{\mG}{\matr{G}}      	 % Greens' kernel

\newcommand{\rmi}{\mathrm{i}}   % Complex number i
\usepackage{array}    % For tables
\usepackage{amsmath,amssymb,bbm,bm}
\usepackage{mathtools}

\usepackage{amsthm}
\theoremstyle{definition}
\newtheorem{definition}{Definition}[section]
\newtheorem{example}{Example}[section]
\newtheorem{remark}{Remark}[section]
\theoremstyle{theorem}

\newtheorem{theorem}{Theorem}[section]

\newtheorem{lemma}{Lemma}[section]
\newtheorem{proposition}{Proposition}[section]
\usepackage{epstopdf}
\DeclareGraphicsExtensions{.pdf,.eps,.png,.jpg,.mps}

\newcommand{\CycleTri}{{\,\atop \mathord{\includegraphics[height=7ex]{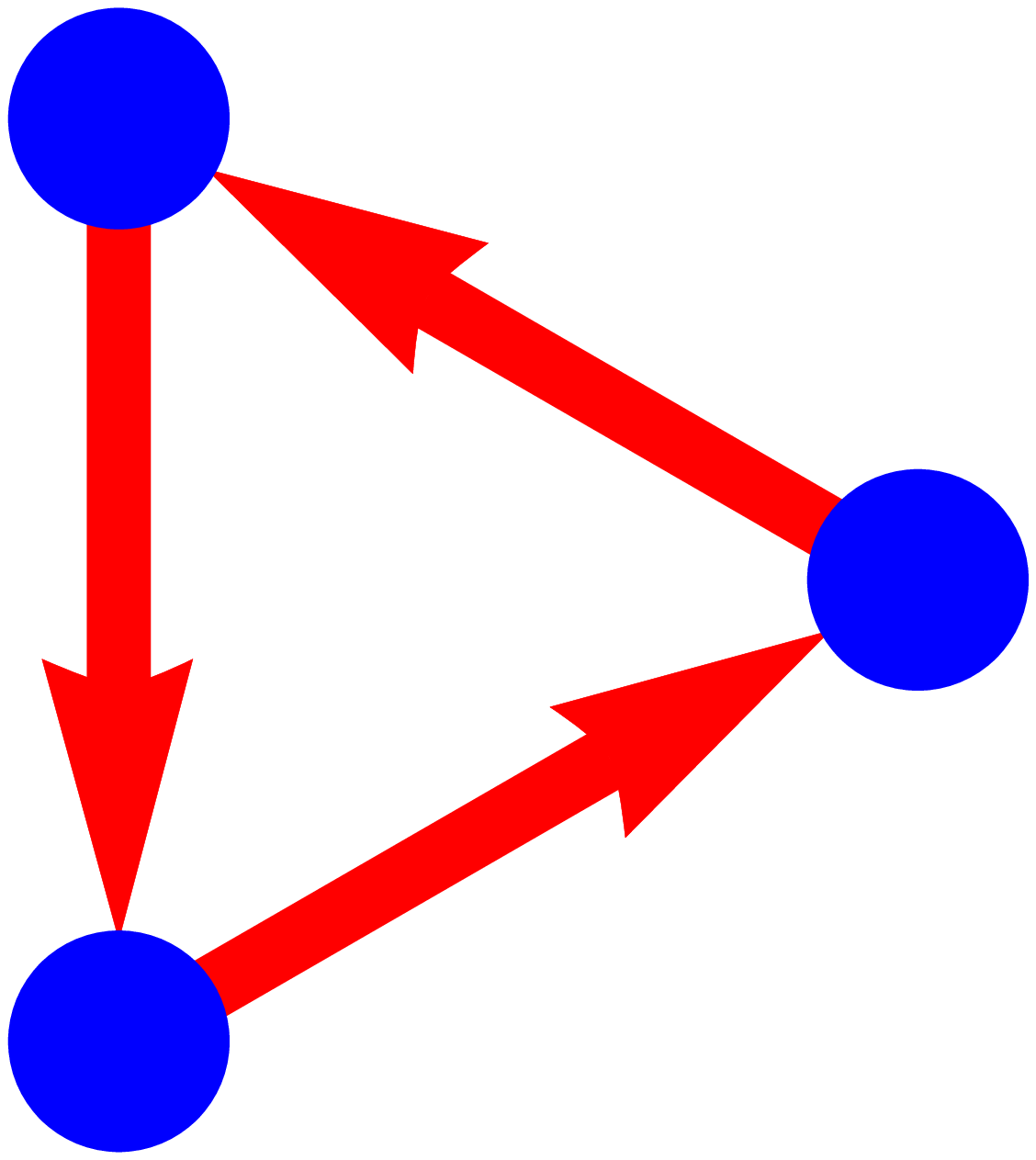}}}}
\newcommand{\PathTwotoOne}{{\,\atop \mathord{\includegraphics[height=7ex]{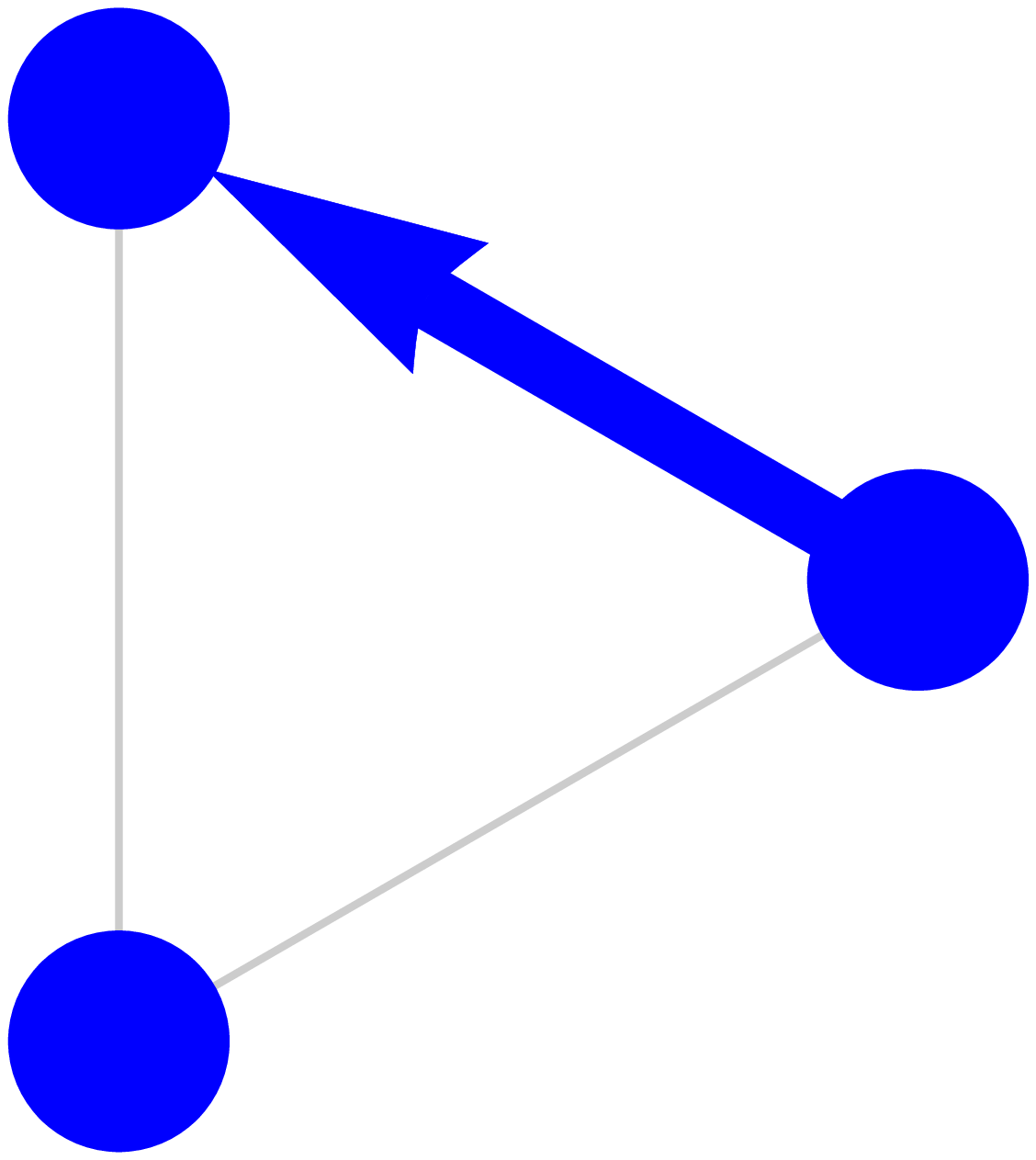}}}}
\newcommand{\PathThreeTwoOne}{{\,\atop \mathord{\includegraphics[height=7ex]{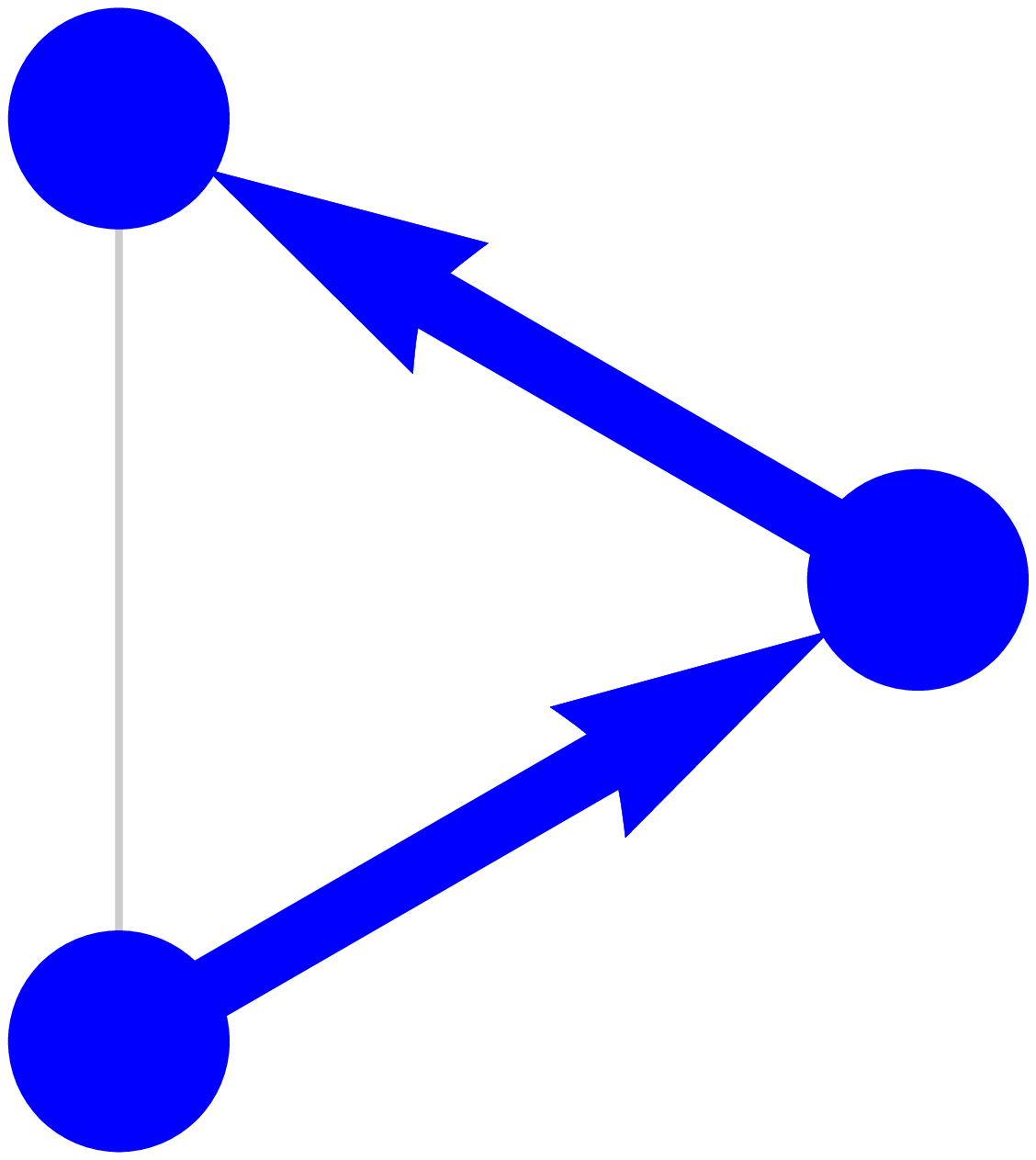}}}}

\newcommand{\PathOneThreeTwo}{{\,\atop \mathord{\includegraphics[height=7ex]{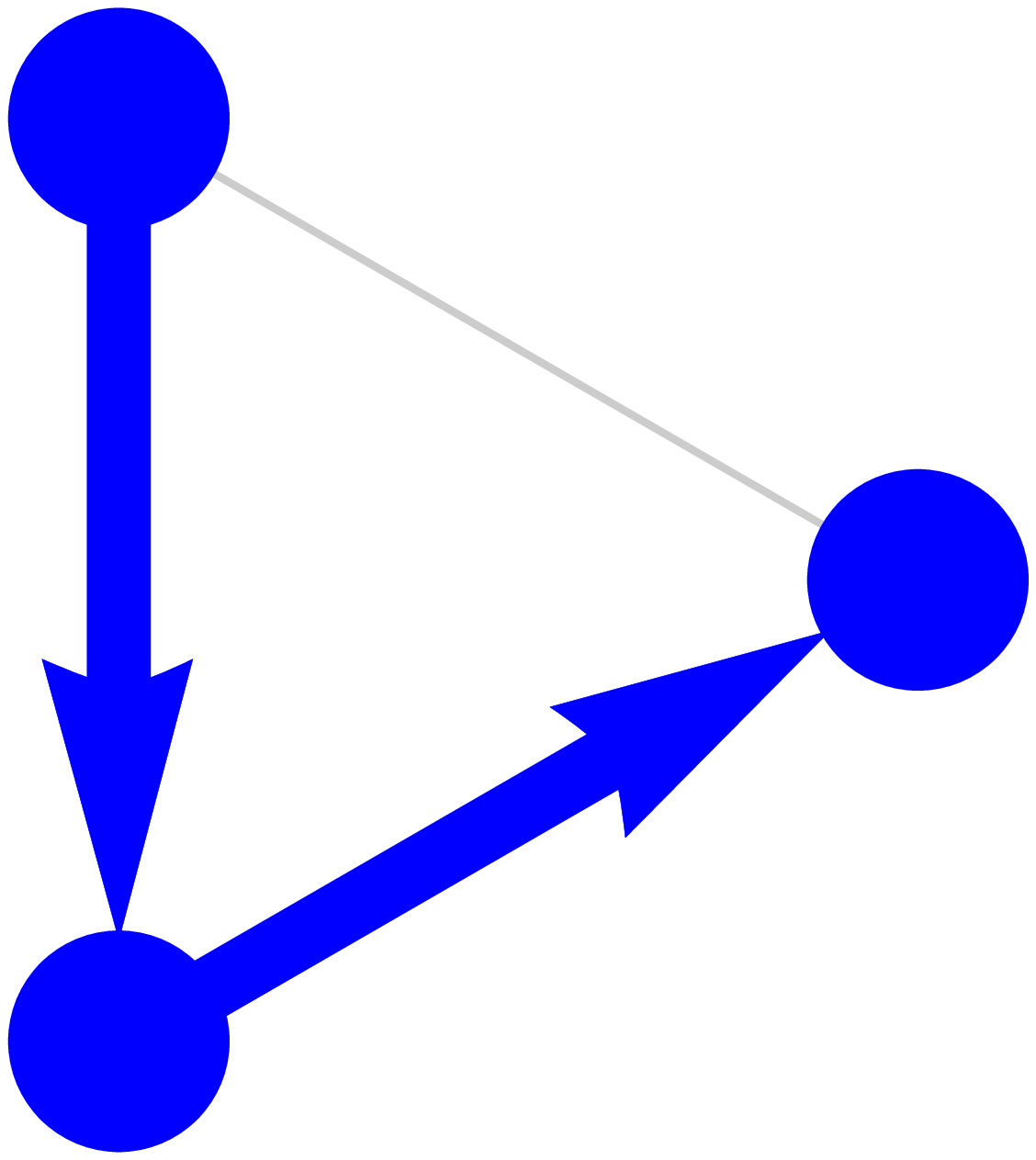}}}}
\newcommand{\PathThreeTwo}{{\,\atop \mathord{\includegraphics[height=7ex]{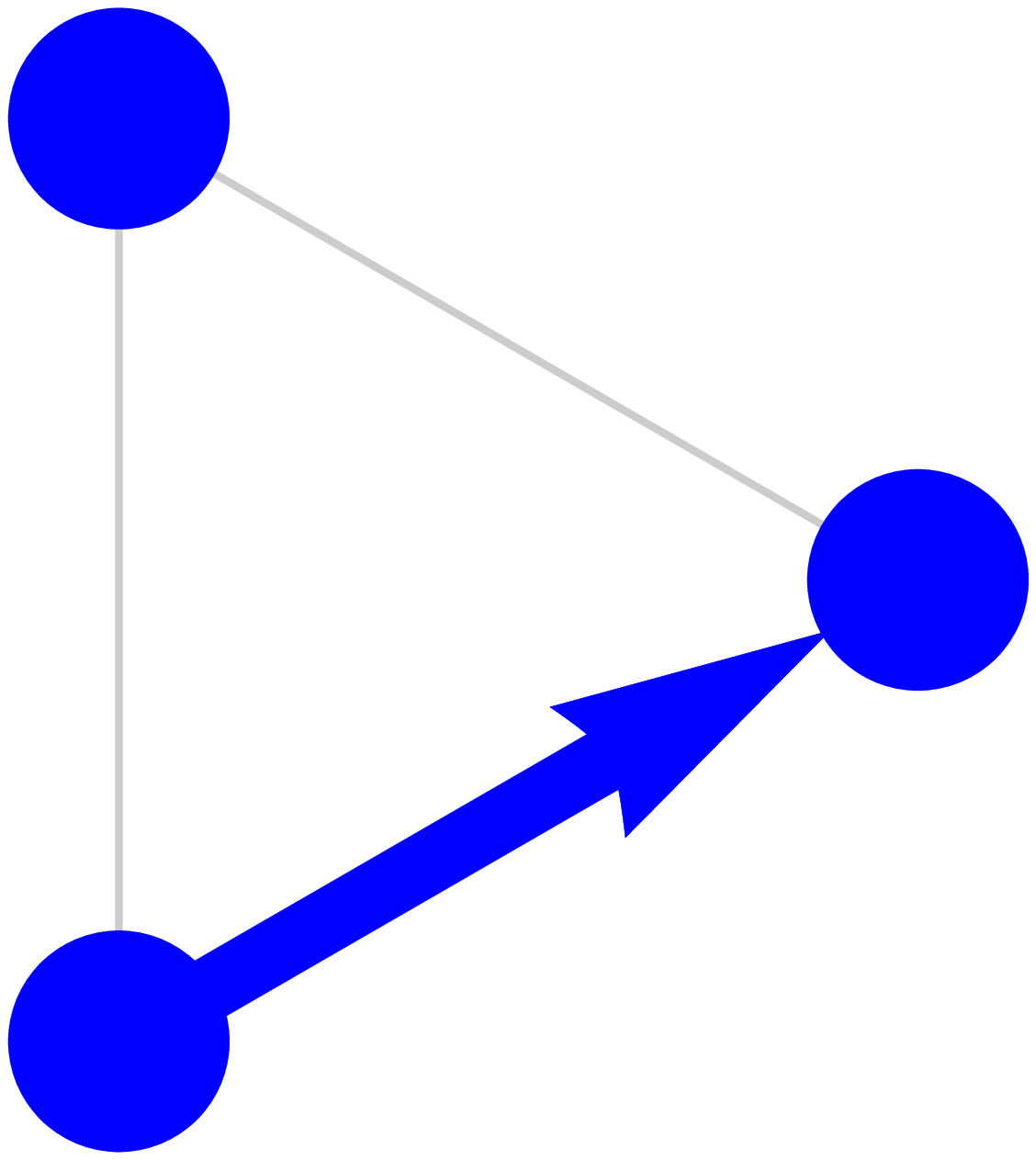}}}}
\newcommand{\PathOneThree}{{\,\atop \mathord{\includegraphics[height=7ex]{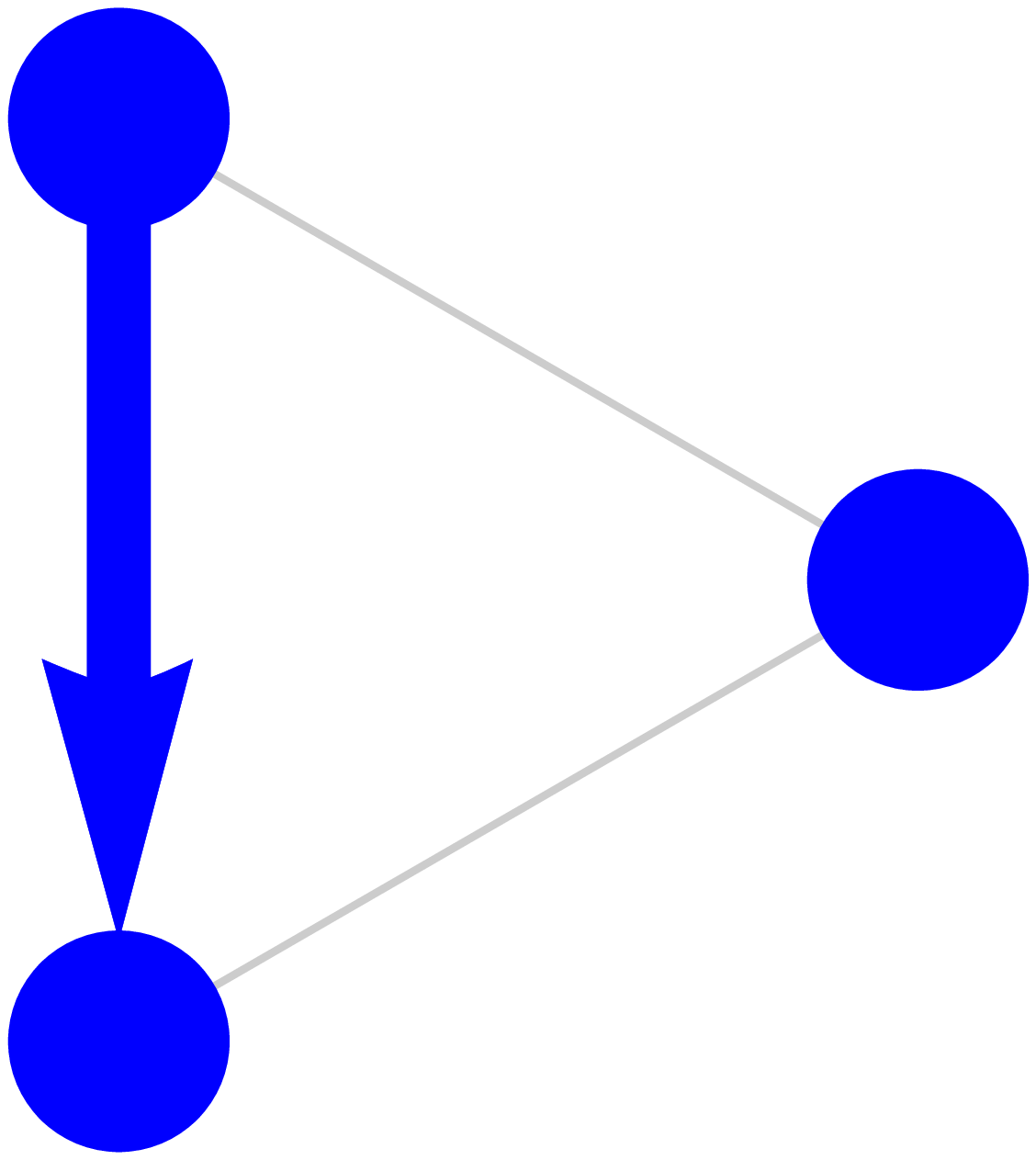}}}}
\newcommand{\PathTwoOneThree}{{\,\atop \mathord{\includegraphics[height=7ex]{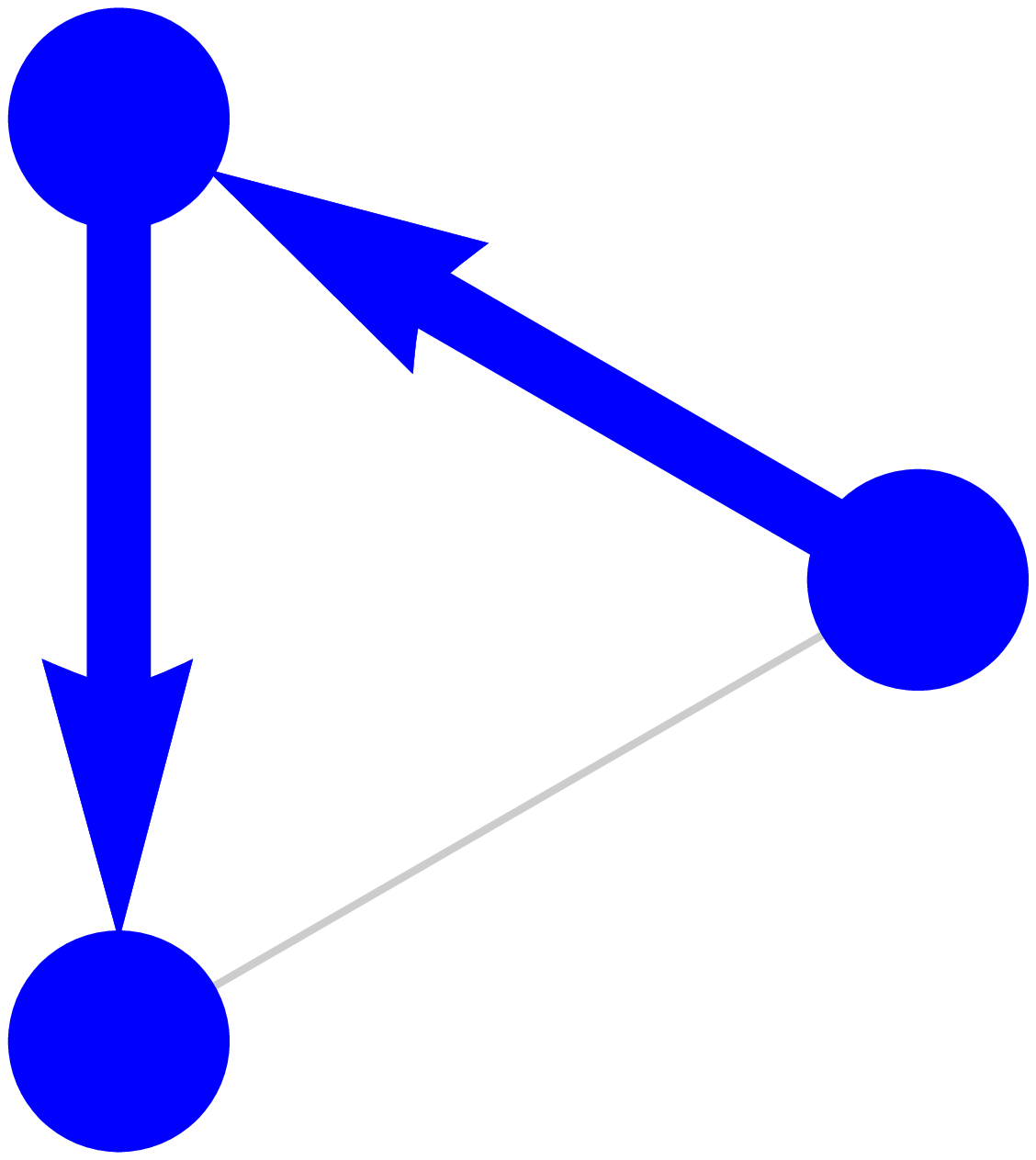}}}}

\newcommand{\BackTrackOneTwo}{{\,\atop \mathord{\includegraphics[height=7ex]{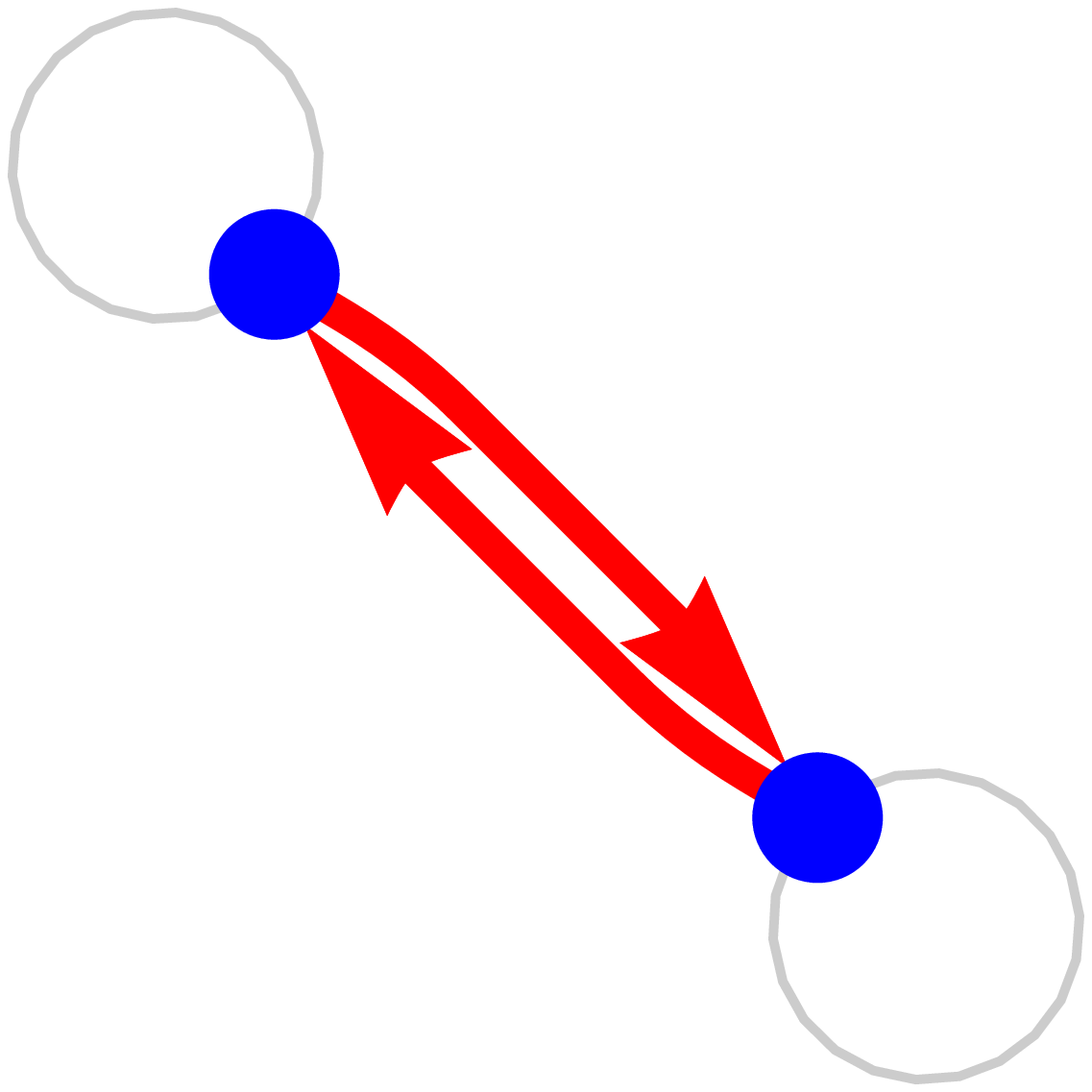}}}}
\newcommand{\BackTrackTwoOne}{{\,\atop \mathord{\includegraphics[height=7ex]{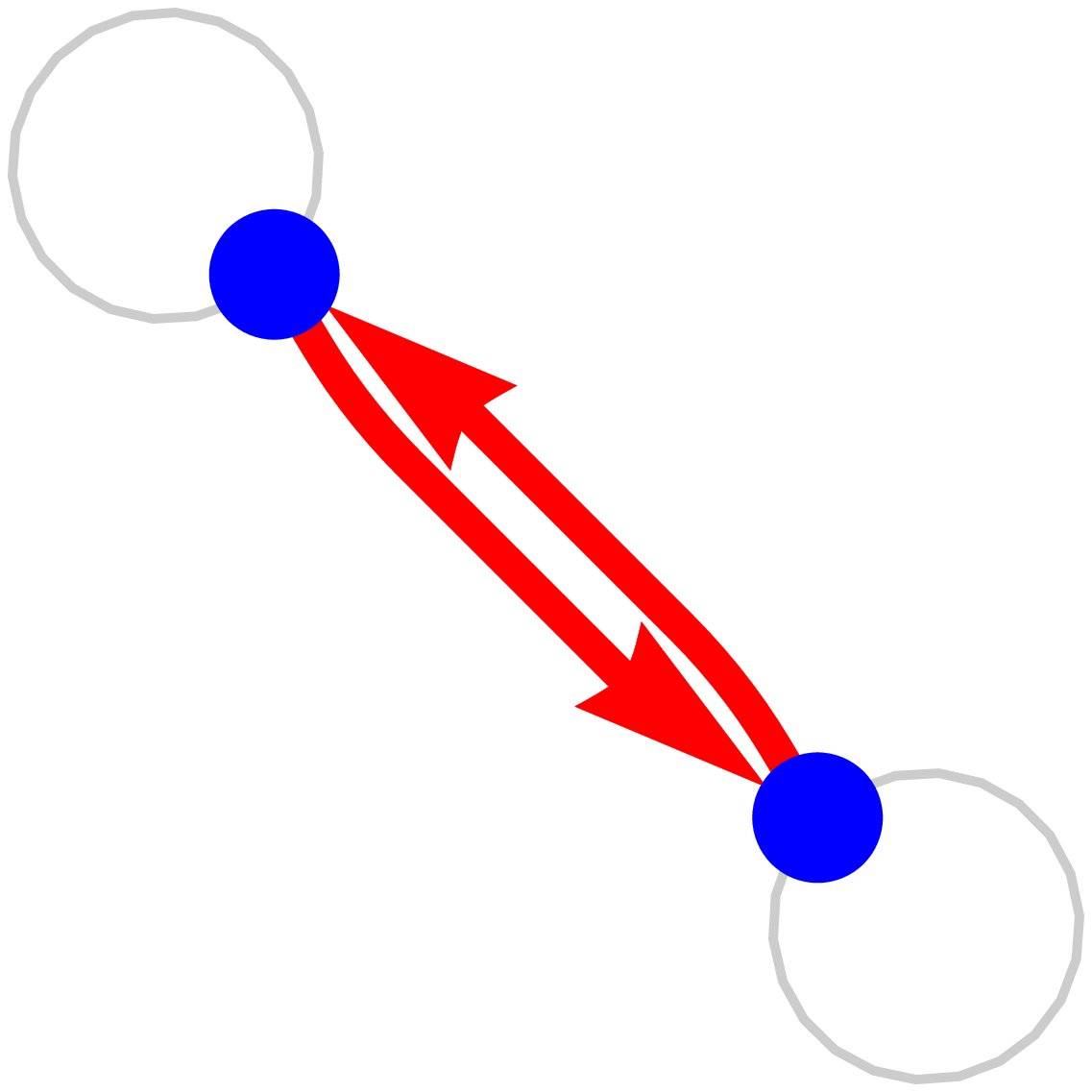}}}}
\newcommand{\LoopOneKTwo}{{\,\atop \mathord{\includegraphics[height=7ex]{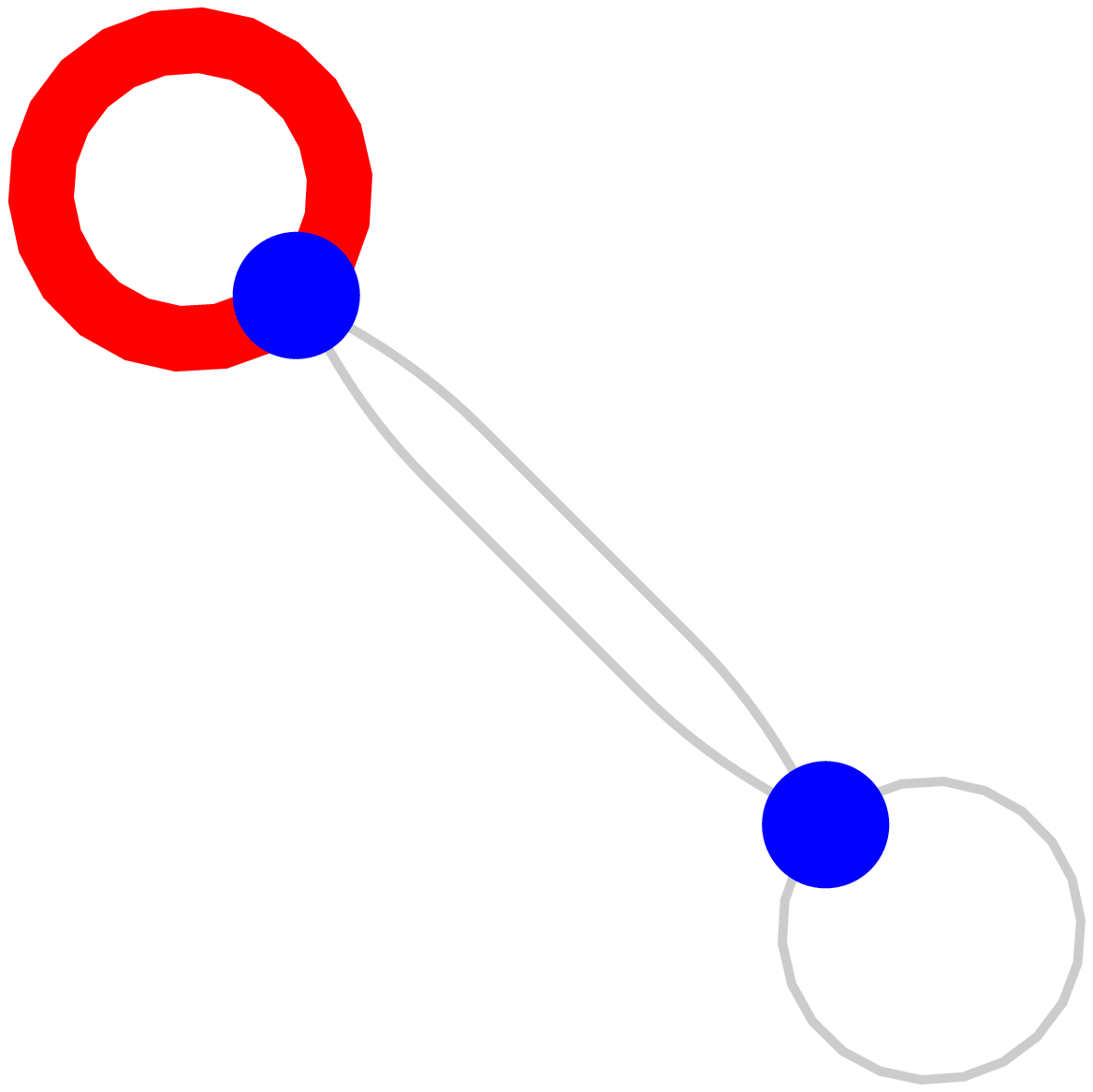}}}}
\newcommand{\LoopTwoKTwo}{{\,\atop \mathord{\includegraphics[height=7ex]{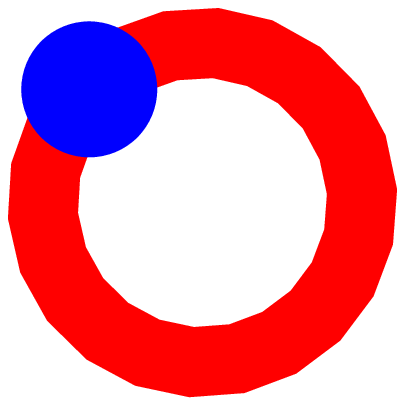}}}}
\newcommand{\LoopTwoTwoKTwo}{{\,\atop \mathord{\includegraphics[height=7ex]{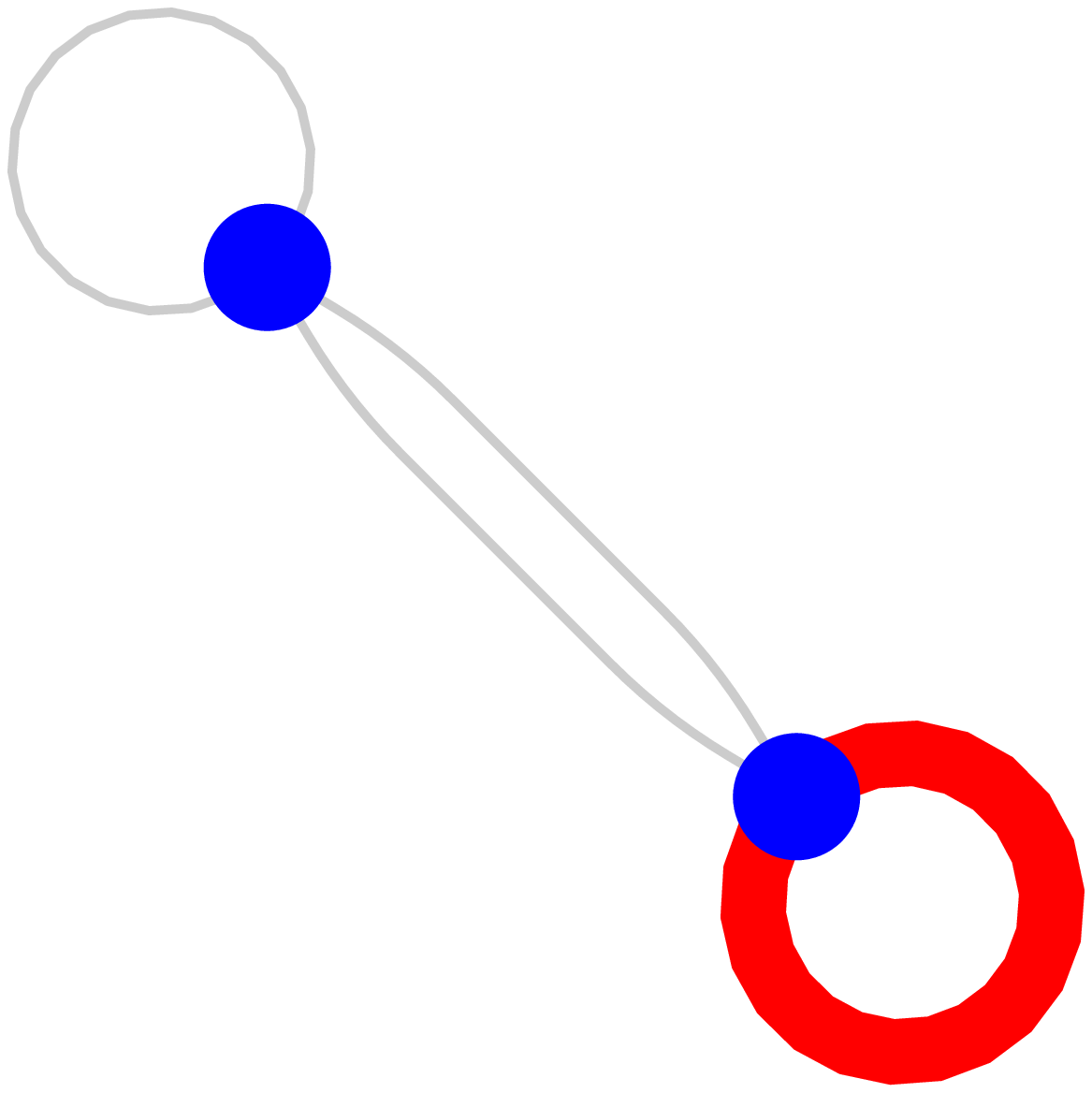}}}}

\newcommand{\PathOneTwo}{{\,\atop \mathord{\includegraphics[height=7ex]{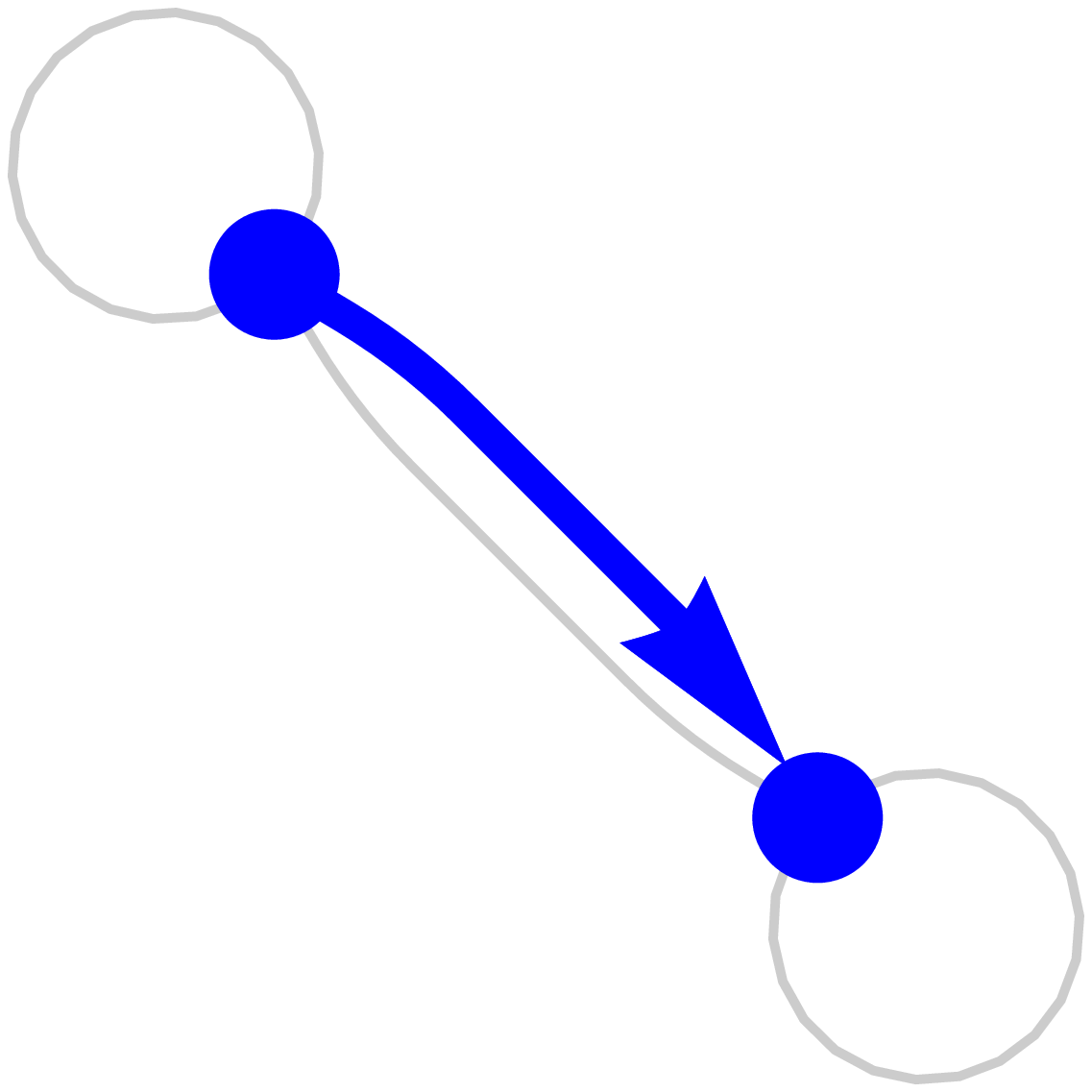}}}}
\newcommand{\PathTwoOne}{{\,\atop \mathord{\includegraphics[height=7ex]{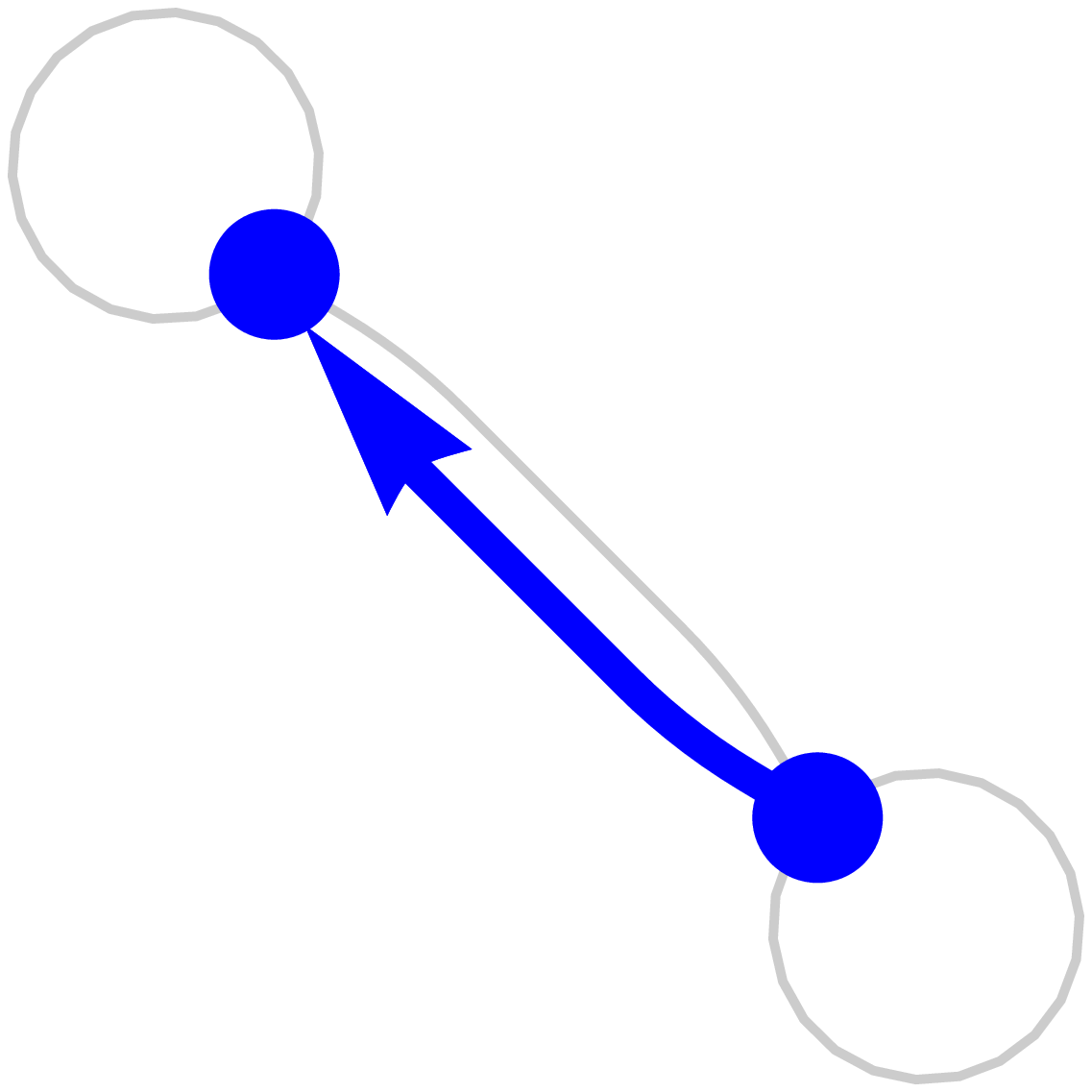}}}}

\journal{Linear Algebra and its Applications}

\begin{document}

\begin{frontmatter}



\title{An Exact Formulation of the Time-Ordered Exponential using Path-Sums}

\author[Oxford]{P.-L. Giscard}
\ead{p.giscard1@physics.ox.ac.uk}
\author[Colby]{K. Lui}
\author[Munich]{S.~J.~Thwaite}
\author[Oxford,Singapore]{D. Jaksch}
\address[Oxford]{University of Oxford, Department of Physics, Clarendon Laboratory, Oxford OX1 3PU, United Kingdom}
\address[Colby]{Colby College, 4000 Mayflower Hill Dr, Waterville, ME 04901, United States}
\address[Munich]{Department of Physics and Arnold Sommerfeld Center for Theoretical Physics, 
Ludwig-Maximilians-Universit\"{a}t M\"{u}nchen, Theresienstra{\ss}e~37, 80333 Munich, Germany}
\address[Singapore]{Centre for Quantum Technologies, National University of Singapore, 3 Science Drive 2, Singapore 117543}

%permits fully analytical calculations of the solution of systems of differential equations with variable coefficients.

\begin{abstract}
We present the path-sum formulation for $\OE[\H](t',t)=\mathcal{T}\,\text{exp}\big(\int_{t}^{t'}\!\H(\tau)\,d\tau\big)$, the time-ordered exponential of a time-dependent matrix $\H(t)$. The path-sum formulation gives $\OE[\H]$ as a branched continued fraction of finite depth and breadth. The terms of the path-sum have an elementary interpretation as self-avoiding walks and self-avoiding polygons on a graph. Our result is based on a representation of the time-ordered exponential as the inverse of an operator, the mapping of this inverse to sums of walks on graphs and the algebraic structure of sets of walks. We give examples demonstrating our approach. We establish a super-exponential decay bound for the magnitude of the entries of the time-ordered exponential of sparse matrices. We give explicit results for matrices with commonly encountered sparse structures.
\end{abstract}

\end{frontmatter}

\section{Context}\label{context}
The time-ordered exponential function $\OE[\H](t',t)=\mathcal{T}\,\text{exp}\big(\int_{t}^{t'}\!\H(\tau)\,d\tau\big)$, also known as path-ordered exponential, is the unique solution of the system of differential equations
\begin{equation}\label{DiffEq}
\H(t')\OE[\H](t',t)=\frac{d}{dt'}\OE[\H](t',t)
\end{equation}
such that $\OE[\H](t',t')=\I$ is the identity at all times.  We take $\H\in\mathbb{C}^{n\times n}[I]$, $n\in\mathbb{N}\backslash\{0\}$, to be a matrix depending smoothly on the continuous variable $t\in I$, which we call time without loss of generality. %The time-ordered exponential function $\OE[.]$ is here defined as the unique solution to this equation, with the further condition that $\OE[\H](t,t)$ be the identity at all times $t\in I$:
%\begin{subequations}
%\begin{align}\label{OEDef}
%&\OE[\H](t',t)v(t)=v(t'),~t\leq t',\\
%&\OE[\H](t,t)v(t)=v(t),
%\end{align}
%\end{subequations}
%where $t,t'\in I$. 
In spite of the importance of the time-ordered exponential as a solution of systems of differential equations with variable coefficients Eq.~(\ref{DiffEq}), it is rarely studied within the same framework as the more common matrix functions. It is for example absent from reference books on matrix functions, e.g. Refs. \cite{Golub,Higham2008}. Rather, many studies concerning the time-ordered exponential find their roots in physics, where it occurs abundantly: for example in non-equilibrium quantum many-body physics and quantum field theories.
Practically speaking, the time-ordered exponential is often calculated perturbatively via the Dyson series or Magnus series and related approaches \cite{Lam1998,Blanes2009}. A non-perturbative result for the time-ordered exponential, known as the time-dependent Dyson equation, also exists. Ultimately however, the terms involved in the Dyson equation (in particular the self-energy) are rarely known explicitly, and the equation is rather used as a starting point for approximations: for example in equilibrium and non-equilibrium Dynamical Mean Field Theory (DMFT) \cite{Georges1996,Aoki2013}.

In this work, we establish a universal non-perturbative formulation for the time-ordered exponential of a time-dependent matrix based on graph theoretic considerations. This result presents the time-ordered exponential as a branched continued fraction of finite depth and breadth, and permits analytical calculations. Our approach is based on three pillars: i) a representation of time-ordered exponentials as inverses of operators; ii) a mapping between such inverses and sums of walks on graphs; and iii) results on the algebraic structure of sets of walks.

This article is organized as follows. In \S\ref{requiredconcept} we introduce a mapping between time-dependent and time-continuous matrices, which we define. Time-continuous matrices incorporate time as a continuous row and column index in addition to discrete indices.
Using these, we reformulate in \S\ref{TimeContinuousInverse} the solution of the system of differential equations of Eq.~(\ref{DiffEq}) as the inverse of the time-continuous matrix constructed from $\I-\H(t)$, $\I$ being the identity. We show that this inverse always exists.
Exploiting the Neumann series representation of this inverse and the correspondence between matrix powers and walks on graphs, we obtain in \S\ref{PathSum} a graph theoretic formulation of the time-ordered exponential function called path-sum: an analytical expression involving only finitely many terms. We then present examples demonstrating our approach. In particular, we show that the time-dependent Dyson equation is the simplest non-trivial instance of the path-sum formula, corresponding to the situation where the underlying graph has only two vertices.
Using the main result of \S\ref{TimeContinuousInverse}, we obtain in \S\ref{decayprop} a super-exponential decay bound for the magnitude of the entries of the time-ordered exponentials of sparse matrices.  We present explicit bounds for commonly encountered sparse structures.
Finally, in \S\ref{discuss} we discuss the future prospects of graph theoretic methods in numerical calculations of the time-ordered exponential function.

\section{Required concepts}\label{requiredconcept}
In this section we establish that the time-ordered exponential arise from the inverse of a time-continuous inverse. In the next section we will exploit this result to obtain a graph theoretic interpretation for $\OE[\H](t',t)$.
%
%A complete presentation of matrices over $C^\ast$-algebras is presented in \cite{BenziManchester}. Here we shall use the equivalence between matrices over a $C^\ast$-algebra and infinite continuous matrices.

\subsection{Time-dependent and time-continuous matrices}
Let $t$ and $t'$ be two real variables, called times for simplicity. We restrict $t$ and $t'$ to a finite interval $I\subseteq\mathbb{R}$. Let $\matr{m}[t',t]$ be a $n\times m$, $n,m\geq1$, complex-valued matrix of functions of \emph{at most two time variables} $t$ and $t'$ in $I$. We say that $\matr{m}[t',t]$ is a \emph{time-dependent matrix}. Unless necessary we shall drop the time-dependencies and write simply $\matr{m}$ for $\matr{m}[t',t]$. We denote by $\matr{m}(t_2,t_1)$ the ordinary matrix obtained upon evaluating $\matr{m}$ at $t'=t_2$ and $t=t_1$.

\vspace{2mm}From now on, we only consider time-dependent matrices such that for all $(t',t)\in I^2$, $\matr{m}(t',t)$ exists and has finite norm. We denote the set of all such matrices $\mathbb{C}^{n\times m}[I^2]$. An addition operation $+$ between matrices of $\mathbb{C}^{n\times m}[I^2]$ is naturally defined element-wise: that is, $\big(\matr{m}_1+\matr{m}_2\big)_{ij}=(\matr{m}_{1})_{ij}+(\matr{m}_2)_{ij}$. A product operation, denoted $\ast$, also exists and will be introduced below.
%\footnote{Just as for any ordinary matrix of $\mathbb{C}^{n\times m}$, we can partition matrices of $\mathbb{C}^{n\times m}[I^2]$ into disjoint blocks (not necessarily contiguous), i.e.~we can form partitions of time-dependent matrices.}

%
% it will also be convenient to consider arbitrary partitions of $\matr{m}$ so that $\matr{m}$ can equally be seen as a matrix over the non-commutative algebra of matrices over the previously mentioned $C^\ast$-algebra.

To any time-dependent matrix $\matr{m}\in\mathbb{C}^{n\times m}[I^2]$ we associate a unique object $\M\in \mathbb{C}^{n\times m}\times I^2$ via the map $\Phi$, which we define as follows:
\begin{subequations}
\begin{align}
\Phi:\quad\mathbb{C}^{n\times m}[I^2]\longrightarrow \,\,&\mathbb{C}^{n\times m}\times I^2,\\
\hspace{5mm}\matr{m}\longrightarrow \,\,&\Phi\big(\matr{m}\big)=\M:=\int_I dt'\! \int_I dt\,\, \matr{m}(t',t)\otimes  e_{t'} e_t^\dagger,\\
\intertext{or, adopting Dirac notation,}
\hspace{5mm}\matr{m}\longrightarrow\,\,&\Phi\big(\matr{m}\big)=\M:=\int_I dt'\! \int_I dt\,\, \matr{m}(t',t)\, |t'\rangle\langle t|.\label{ContFuncDecompo}
\end{align}
\end{subequations}
In these expressions, $e_{t'}\equiv |t'\rangle$ is a continuous vector identically 0 at all times except at $t'$ where it is infinite, $e_{t}^\dagger\equiv\langle t| $ is the dual of $e_{t}\equiv |t\rangle$, i.e.~its conjugate transpose and $e_{t}^\dagger e_{t'}=\delta(t-t')$, the Dirac delta function. Equation (\ref{ContFuncDecompo}) is the continuous analog of the representation of discrete matrices in terms of their entries, e.g. $\matr{B}=\sum_{i,j} b_{i,j}\, \e_i^\dagger \e_j\equiv\sum_{i,j}b_{i,j}|i\rangle\langle j|$.

We call the elements of $\mathbb{C}^{n\times m}\times I^2$ such as $\M$ the \emph{time-continuous matrices}. Since the map $\Phi$ is clearly bijective, with $\Phi^{-1}(\M)=\matr{m}$, time-continuous matrices are in one-to-one correspondence with time-dependent matrices. In addition, $\Phi$ is linear so $\Phi\big(\matr{m}_1+\matr{m}_2\big)=\Phi\big(\matr{m}_1\big)+\Phi\big(\matr{m}_2\big)$, $(\matr{m}_1,\matr{m}_2)\in\big(\mathbb{C}^{n\times m}[I^2]\big)^2$. Below we also construct a product operation between time-dependent matrices.

%We denote by $.$ the multiplication operation between operators of $\mathbb{C}^{n\times n}\times I^2$. In direct analogy with finite matrices, the appropriate product for operators is found to be
%\begin{equation}
%\int_{t}^{t'}\!\matr{m}_2(t',\tau).\matr{m}_1(\tau,t)\,d\tau.
%
%\end{equation}

\subsection{Product operation for time-dependent matrices}
We define a product operation between time-dependent matrices product by requiring $\Phi$ be an homomorphism. Since $\Phi$ is bijective we can write
\begin{align}
\ast:\quad\mathbb{C}^{n\times m}[I^2]\times \mathbb{C}^{m\times p}[I^2]\longrightarrow \,\,&\mathbb{C}^{n\times p}[I^2],\nonumber\\
\hspace{3mm}(\matr{m}_2,\,\matr{m}_1)\longrightarrow \,\,&\matr{m}_2\ast\matr{m}_1:=\Phi^{-1}\big(\Phi(\matr{m}_2)\,.\,\Phi(\matr{m}_1)\big),\label{line2}
\end{align}
and indeed $\Phi(\matr{m}_2\ast\matr{m}_1)=\Phi(\matr{m}_2)\,.\,\Phi(\matr{m}_1)$.
In Eq.~(\ref{line2}), the $.$ symbol represents the usual matrix product for time-continuous matrices, which is readily obtained as the continuous analog of the matrix product between discrete matrices:
\begin{equation}\label{eq:TimeContinuousProduct}
\Phi(\matr{m}_2)\,.\,\Phi(\matr{m}_1)=\int_I dt'\! \int_I dt\,\, \left(\int_{t}^{t'}\!d\tau\,\matr{m}_2(t',\tau)\matr{m}_1(\tau,t)\right)\otimes  e_{t'} e_t^\dagger,
\end{equation}
where $\matr{m}_2(t',\tau)\matr{m}_1(\tau,t)$ denotes the usual product between finite discrete matrices. For clarity, we omit the $.$ symbol from now on.
Eqs.~(\ref{line2}, \ref{eq:TimeContinuousProduct}) lead to
\begin{equation}\label{astproddef}
\big(\matr{m}_2\ast\matr{m}_1\big)(t',t)=\int_{t}^{t'}\!\matr{m}_2(t',\tau)\matr{m}_1(\tau,t)\,d\tau.
\end{equation}
Although Eq.~(\ref{astproddef}) is reminiscent of a convolution, in general the $\ast$-product is not a convolution since neither $\matr{m}_1$ nor $\matr{m}_2$ need be functions of the time differences $\tau-t$ and $t'-\tau$, respectively. Eq.~(\ref{astproddef}) indicates that the $\ast$-product is a usual matrix multiplication along the (continuous) time index. In Appendix \ref{StarPowersInv} we give the explicit expression of $\matr{m}^{\ast n}$, the $n^{\text{th}}$ $\ast$-power of a time-dependent matrix $\matr{m}$, as well as a characterization of its $\ast$-inverse $\matr{m}^{\ast -1}$.\\

\noindent We gather our results concerning $\Phi$ in a proposition:
\begin{proposition}
The map $\Phi$ forms a ring isomorphism between $(\mathbb{C}^{n\times n}[I^2],+,\ast)$ and $(\mathbb{C}^{n\times n}\times I^2,+,.)$.
\end{proposition}

\begin{proof}
Clearly both $(\mathbb{C}^{n\times n}[I^2],+,\ast)$ and $(\mathbb{C}^{n\times n}\times I^2,+,.)$ are rings. We have already proven that $\Phi$ is a linear bijective homomorphism and thus we only need to verify that it maps the time-dependent identity to the time-continuous identity. This is straightforward on noting that the time-dependent identity associated with the $\ast$-product is $1_\ast:=\delta(t'-t)\I_n$, with $\delta(t'-t)$ the Dirac delta function. Then, since $\Phi$ is an homomorphism, $\Phi(1_\ast)$ is the time-continuous identity.
%The result of the Proposition can be recovered upon noting that once equipped with the conjugate-transpose operation (symbol $\dagger$), both $(\mathbb{C}^{n\times n}[I^2],+,\ast,\dagger)$ and $(\mathbb{C}^{n\times n}\times I^2,+,.,\dagger)$ form $C^\ast$-algebras of dimension $n$ (which are noncommutative when $n\geq 2$) and hence are isomorphic. Here we have explicitly constructed this isomorphism.
\end{proof}

In the following section we exploit the isomorphism $\Phi$ to translate questions involving time-dependent matrices into questions involving time-continuous ones. This leads to a surprisingly simple formulation for the solution of systems of linear differential equations with variable coefficients.

%In other terms, $\matr{f}^{\ast-1}$ is the (matrix) Green's function of a Volterra equation with integral kernel $\matr{f}$. Naturally, the $\ast$-inverse is simply the operator inverse along the continuous time index, indeed if we define
%\begin{equation}
%\matr{F}^{-1}:=\int_I\!dt'\int_I\!dt\,\,\matr{f}^{\ast-1}(t',t)\otimes e_{t'}e^{\dagger}_t,
%\end{equation}
%then we have $\matr{F}\matr{F}^{-1}=\matr{F}^{-1}\matr{F}=\I$. It must be noted that from this equation the $\ast$-inverse of $\I_n-\matr{f}$ has a series representation,
%\begin{subequations}
%\begin{equation}\label{InversePowerSeries}
%(\I_n-\matr{f})^{\ast-1}=\sum_{n=0}^{\infty}\matr{f}^{\ast n},
%\end{equation}
%or explicitely
%\begin{align}
%&\hspace{-1mm}(\I_n-\matr{f})^{\ast-1}(t',t)=\\
%&\hspace{2mm}\sum_{n=0}^{\infty}\int_{t}^{t'}\int_{t}^{\tau_{n-1}}\hspace{-3mm}\cdots\int_{t}^{\tau_2}\matr{f}(t',\tau_{n-1})\matr{f}(\tau_{n-1},\tau_{n-2})\cdots \matr{f}(\tau_2,\tau_1)\matr{f}(\tau_1,t)\,d\tau_1\cdots d\tau_{n-1}.\nonumber
%\end{align}
%\end{subequations}

\section{The time-ordered exponential is a time-continuous inverse}\label{TimeContinuousInverse}
We define four time-continuous matrices related to the system of differential equations Eq.~(\ref{DiffEq}): 
\begin{definition}\label{DefTimeCont}
Let $\I_n\in\mathbb{C}^{n\times n}$, $n\geq 0$, be the $n\times n$ identity matrix.
Let $\matr{H}(t)\in\mathbb{C}^{n\times n}[I^2]$ depend on \emph{one} time variable.
We define
\begin{align*}
%&\v:=\Phi\big(v(t)\big),\hspace{33.3mm}\text{Time-continuous vector},\\
&\I:=\Phi\Big(\delta(t'-t)\,\I_n\Big),\!\hspace{23.65mm}\text{Time-continuous identity},\\
&\C:=\Phi\Big(\Theta(t'-t)\,\I_n\Big),\hspace{20.7mm}\text{Time-continous causality matrix},\\
&\A:=\Phi\Big(\Theta(t'-t)\,\H(t)\Big),\hspace{17mm}\text{Time-continous motion generator},\\
&\T:=\Phi\Big(\Theta(t'-t)\, \OE[\H](t',t)\Big),\hspace{6mm}\text{Time-continous propagator,}
\end{align*}
In these expressions, $\delta$ is the Dirac delta function and $\Theta(x)$ is the Heaviside step function, with the convention that $\Theta(0)=1$, that is
\begin{equation}\label{Heaviside}
\Theta(x):=\begin{cases}1&\text{if }x\geq0,\\
0&\text{otherwise.}\end{cases}
\end{equation}
\end{definition}
The relation between the time-continuous motion generator and propagator is the main result of this section:

\begin{theorem}\label{MainRes}
Let $I\subset\mathbb{R}$ be a finite interval and $\H(t)\in \mathbb{C}^{n\times n}[I^2]$. Then $\I-\A\in\mathbb{C}^{n\times n}\times I^2$ is invertible and the system of linear differential equations $\H(t) v(t)=\dot{v}(t)$ has time-continuous propagator
\begin{equation}
\T=(\I-\A)^{-1}\C,\label{TCProp}
\end{equation}
or equivalently
\begin{subequations}\label{OEThm}
\begin{align}
\Theta(t'-t)\OE[\H](t',t)&=\Phi^{-1}\big((\I-\A)^{-1}\C\big)(t',t),\label{OEThm1}\\
%&=\big(\I_n \otimes e_{t'}^\dagger\,\big)\,(\I-\A)^{-1}\M\,\big(\I_n \otimes e_{t}\big),\\
&=\int_{t}^{t'} \Phi^{-1}\big((\I-\A)^{-1}\big)(t',\tau)\,\,d\tau,\label{OEThm2}
\end{align}
\end{subequations}
for any $(t',t)\in I^2$.\\
%
%=\int_{t}^{t'}\mG(\tau,t)\,d\tau.
%\end{equation}
\end{theorem}

\begin{proof}[Proof of Theorem \ref{MainRes}]
We prove the invertibility of $\I-\A$. The following Lemma will be necessary:

\begin{lemma}\label{Mpowerm}
Let $\C$ be the time-continuous causality matrix on $I$. Then
%\begin{subequations}
\begin{equation}
\Phi^{-1}\big(\C^{m}\big)=\Theta(t'-t)\frac{(t'-t)^{m-1}}{(m-1)!},
\end{equation}
for $t'\geq t$ on $I$.
%and
%\begin{equation}
%\I_n\otimes\M^{-1}=\d.
%\end{equation}
%\end{subequations}
\end{lemma}
\begin{proof}
The lemma
%concerning $e_{t'}^{\dagger}\M^{m}e_t$
follows from an induction. The result holds trivially  for $m=1$. Supposing that it holds up to $m$ and using the fact that $\Phi$ is a ring homomorphism, we have
\begin{align}
\Phi^{-1}\big(\C^{m+1}\big)&=\Phi^{-1}\big(\C\big)\ast\Phi^{-1}\big(\C^{m+1}\big),\\
&=\int_I\!d\tau \,\,\Theta(t'-\tau)\Theta(\tau-t)\frac{(\tau-t)^{m-1}}{(m-1)!},\nonumber\\
&=\Theta(t'-t)\frac{(t'-t)^{m}}{m!},
\end{align}% \Theta(t'-t)\int_{t}^{t'} \!d\tau \,\,\frac{(\tau-t)^{m-1}}{(m-1)!}
which establishes the
%first result of the
result.
%For the second result
%\begin{align*}
%\d\M\v&=\int_I\!dt'\,\,\frac{d}{dt} \delta(t'-t)e_{t'}e^\dagger_t  \int_I dt_2\int_I dt_1\, \Theta(t_2-t_1)|v(t_1)\rangle|t_2\rangle,\\
%&=\int_I dt \int_I dt_1 \,\frac{d}{dt} \Theta(t-t_1)|v(t_1)\rangle|t\rangle,\\
%&=\int_I dt \int_I dt_1 \,\delta(t-t_1)|v(t_1)\rangle|t\rangle,\\
%&=\int_I dt \,|v(t)\rangle|t\rangle,\\
%&=\v.
%\end{align*}
%Similarly we show $\M\d\v=\v$.
\end{proof}
\noindent Consider the series
\begin{equation}
S_K[\A](t',t):=e_{t'}^{\dagger}\left(\sum_{k=0}^K \A^{k}\right)e_{t},
\end{equation}
with $t,t'\in I$ and $K\geq 0$ finite. To prove the invertibility of $\I-\A$ on $I$ we obtain a finite upper bound for $S_K[\A](t',t)$ for all $K\geq0$.
First, if $t'<t$, $\Theta(t'-t)=0$ and thus $S_K[\A](t',t)=0$ is trivially bounded for all $K\geq0$. Otherwise, for $t'\geq t$,
\begin{align}
\|S_K[\A](t',t)\|&\leq \sum_{k=0}^K \|e_{t'}^\dagger \A^ke_t\|,\nonumber\\
&\leq\sum_{k=0}^K\big\| \|\H\|^k_{\infty,\,I} \,\,e_{t'}^\dagger\C^ke_t\big\|.
\end{align}
Using Lemma \ref{Mpowerm} for the norm of $\|e_{t'}^\dagger \C^k e_t\|$ yields
%\begin{subequations}
\begin{align}
\|S_K[\A](t',t)\|&\leq \sum_{k=0}^K \|\H\|^k_{\infty,\,I} \,\frac{(t'-t)^{k-1}}{(k-1)!},\nonumber\\
&=\|\H\|_{\infty,\,I}\, e^{\|\H\|_{\infty,\,I} (t'-t)} \,\frac{\Gamma\big(K,\|\H\|_{\infty,\,I} (t'-t)\big)}{(K-1)!},\nonumber\\
&\leq \|\H\|_{\infty,\,I}\, e^{\|\H\|_{\infty,\,I} (t'-t)},\label{Sklimit}
\end{align}
%\end{subequations}
where $\Gamma(a,b)=\int_{b}^\infty t^{a-1}e^{-t}dt$ is the upper incomplete gamma function and $\|\H\|_{\infty,\,I}=\sup_{t\in I}\|\H(t)\|$ is finite since $\H(t)\in\mathbb{C}^{n\times n}[I^2]$. %The last upper bound holds for any $K$
Furthermore, by Lemma \ref{Mpowerm} we have
\begin{equation}
\lim_{K\to\infty} \|S_K[\A](t',t)-S_{K-1}[\A](t',t)\|=\lim_{K\to\infty}\|e_{t'}^\dagger \A^K e_t\|=0.
\end{equation}
Therefore, the limit $S_\infty[\A](t',t):=\lim_{K\to\infty}S_K[\A](t',t)$ exists and, by Eq.~(\ref{Sklimit}) and Fatou's lemma, is bounded by
\begin{equation}
S_\infty[\A](t',t)\leq \|\H\|_{\infty,\,I} \,e^{\|\H\|_{\infty,\,I}\, (t'-t)}.
\end{equation}
We have shown that on any finite time interval $[t,t']\subseteq I$, the series $\sum_{k=0}^\infty \A^k$ is absolutely convergent. Hence $\I-\A$ is invertible and $e_{t'}^{\dagger}(\I-\A)^{-1}e_t=S_\infty[\A](t',t)$. Remarkably, as long as $\|\H\|_{\infty,\,I}$ is finite, this holds regardless of the details of $\H(t)$. In particular $e_{t'}^{\dagger}(\I-\A)^{-1}e_t$ exists even when $\I_n-\H(t)$ is not invertible for some or all times $t$ in $I$.
Ultimately this is because the $\ast$-inverse of a time-dependent matrix always exists: for example the $\ast$-inverse of the matrix which is identically zero at all times is the $\ast$-identity denoted $1_\ast:=\delta(t'-t) \matr{I}_n$. This is verified from the Volterra equation satisfied by $(1_\ast-1_\ast)^{\ast-1}-1_\ast$ (see Eq.~(\ref{VolterraInverse}) of Appendix \ref{AstInverse}.)
These results are perhaps less surprising given Lemma \ref{Mpowerm}, which shows that, through powers of $\C$, the inverse of $\I-\A$ gives rise to exponential functions of $\H$. Such functions exist regardless of the invertibility of $\H$.

We turn to proving the formula for the time-continuous propagator $\T$.
Let $v(t)\in\mathbb{C}^{n\times 1}[I]$ be a time-dependent vector such that for all $t\in I$, $v(t)$ satisfies $\H(t)v(t)=\dot{v}(t)$. Let  $\v:=\Phi\big(v(t)\big)$ be the corresponding time-continuous vector.
Then the result for $\T$ follows from a direct calculation of $\T\v-\A\T\v$. We have
\begin{align}
\T \v&=\int_I dt'\! \int_I d t\,\, \Theta(t'-t)\, \OE[\H](t',t)\,v(t) \otimes e_{t'}\nonumber\\
&=\int_I dt'\! \int_I d t\,\, \Theta(t'-t)\,v(t') \otimes e_{t'}
\end{align}
and further
\begin{align}
\A\T &=\int_I dt'\! \int_I d\tau\!\int_I  dt\,\, \Theta(t'-\tau)\Theta(\tau-t) \H(\tau)\OE[\H](\tau,t)\otimes e_{t'}e_t^{\dagger},\nonumber\\
\shortintertext{whence}
\A\T\v &=\int_I dt'\! \int_I d\tau\!\int_I  dt\,\,\Theta(t'-\tau)\Theta(\tau-t) \H(\tau)\OE[\H](\tau,t)v(t)\otimes e_{t'}\nonumber\\
&=\int_I dt'\! \int_I d\tau\!\int_I  dt\,\, \Theta(t'-\tau)\Theta(\tau-t) \H(\tau)v(\tau)\otimes e_{t'}\nonumber\\
&=\int_I dt'\! \int_I d\tau\!\int_I  dt\,\, \Theta(t'-\tau)\Theta(\tau-t) \dot{v}(\tau)\otimes e_{t'}\nonumber\\
&=\int_I dt'\! \int_I dt\, \,\Theta(t'-t) \int_{t}^{t'}\! d\tau\,\,\dot{v}(\tau)\otimes e_{t'}\nonumber\\
&=\int_I dt'\! \int_I dt\,\, \Theta(t'-t) \big( v(t')-v(t)\big)\otimes e_{t'}.
\end{align}
where the second equality has been obtained by substituting $\OE[\H](\tau,t)v(t)=v(\tau)$, since $v(t)$ is a solution for all $t\in I$. It follows that $\T\v-\A\T\v=\int_I \!dt' \int_I\! dt\,\, \Theta(t'-t) v(t)\otimes e_{t'}$.
Finally, we have
\begin{align}
\C\v=\int_I \!dt' \int_I\! dt\,\, \Theta(t'-t) v(t)\otimes e_{t'}=\T\v-\A\T\v.
\end{align}
Since this holds for all time-continuous vectors $\v$ that are solutions to the matrix differential equation, and since $\I-\A$ is invertible, we have $\T=(\I-\A)^{-1}\C$, which proves Eq.~(\ref{TCProp}).

We prove Eqs.~(\ref{OEThm1}, \ref{OEThm2}). Since $\Theta(t'-t)\OE[\H](t',t)=\Phi^{-1}(\T)$ by the definition of $\T$, Eq.~(\ref{OEThm1}) follows immediately from Eq.~(\ref{TCProp}). For Eq.~(\ref{OEThm2}), define $\matr{G}(t',t)$ by $\Theta(t'-t)\matr{G}(t',t):=\Phi^{-1}\big((\I-\A)^{-1}\big)$. Then
\begin{align}
\Phi^{-1}(\T) &= \Phi^{-1}\big((\I-\A)^{-1}\C\big),\nonumber\\
&=\Phi^{-1}\left(\int_{I}dt'\int_I d\tau' \Theta(t'-\tau')\matr{G}(t',\tau') \otimes e_{t'}e_{\tau'}^\dagger\,\int_{I}d\tau\int_Idt\,\Theta(\tau-t)\I_n\otimes e_{\tau}e_{t}^\dagger\right),\nonumber\\
&=\Phi^{-1}\left(\int_{I}dt'\int_I dt\, \Theta(t'-t) \left(\int_{t}^{t'}d\tau\,\matr{G}(t',\tau)\right) \otimes e_{t'}e_{t}^\dagger\right),\nonumber\\
&=\Theta(t'-t)\,\int_{t}^{t'}d\tau\,\matr{G}(t',\tau).
\end{align}
This proves the Theorem.
\end{proof}

\begin{remark}[Ordered exponential of time-independent matrix]
In the situation where $\H(t)\equiv \H_0$ is time independent, it is well known that  the time-ordered exponential is simply a matrix exponential; that is, for $t\leq t'$ we have
\begin{equation}
\OE[\H](t',t)=e^{\H_0 (t'-t)}.
\end{equation}
This is recovered from Theorem \ref{MainRes} on noting that in this situation $\A=\Phi\big(\Theta(t'-t)\H_0\big)=\C\Phi(\H_0)$. Hence, using Proposition \ref{Mpowerm}, we have
\begin{align}
\OE[\H_0](t',t)&=\Phi^{-1}\left((\I-\A)^{-1}\C \right),\nonumber\\
&=\Phi^{-1}\left(\sum_{m=0}^{\infty} \Phi(\H_0)^m \,\C^{m+1}\right),\\
\shortintertext{Since $\Phi$ is a homomorphism, $\Phi(\H_0)^m=\Phi(\H_0^m)$, and Lemma \ref{Mpowerm} leads to}
&=\sum_{m=0}^{\infty}\H_0^m\,\Theta(t'-t)\frac{(t'-t)^m}{m!}\nonumber\\
&=\Theta(t'-t)\,e^{\H_0(t'-t)}.
\end{align}
A consequence of this result is that the matrix exponential of any $n\times n$ complex matrix is a submatrix of the inverse of a time-continuous matrix.
%This statement is not to be confused with the relation between the exponential and resolvent of a matrix via the Cauchy integral formula.
\end{remark}

\section{Path-sum formulation of the time-ordered exponential function}\label{PathSum}
%Theorem \ref{MainRes} shows that the time-continuous propagator $\T$ is proportional to the inverse of $\I-\A$. Therefore, one can obtain or approximate $\T$ with any technique amenable to calculating an inverse. For example, discretizing time, $\I-\A$ becomes an ordinary matrix and the problem of finding $\T$ can be tackled using standard techniques such as LU decomposition or Gaussian elimination. Here we present an alternative approach, in which time is not discretized, but rather $\H$ is interpreted as the weighted adjacency matrix of a graph. This permits the calculation of $(\I-\A)^{-1}$ using a technique from graph theory. This approach yields exact analytical results but can be difficult to use in practice, due to the presence of $\ast$-inverses.

\subsection{Context}
%The \textit{``most basic result of algebraic graph theory"} (as described in \citealp{Flajolet2009}) states that the powers of the adjacency matrix $A_\G$ of a graph $\G$ generate all the walks on this graph \citep{Biggs1993}. This result extends to weighted graphs by constructing a weighted adjacency matrix $R$, with $R_{ij}$ the weight of the edge from vertex $j$ to vertex $i$ on $\G$. Then $(R^\ell)_{ij}$ is the sum of the weights of all the walks of length $\ell$ from $j$ to $i$ \citep{Flajolet2009}. The weight of a walk is simply the ordered product of the weight of the edges it traverses. Note how the indices of a matrix are written right-to-left and correspond to an edge written left-to-right. This is because of unfortunate conventions. We index the entries of $R$ by using the labels of the vertices in $\G$. We write these labels with roman letters, numerals or Greek letters, as convenient.
%Now assuming convergence, since $(I-R)^{-1}=\sum_{\ell\geq0} R^\ell$, it follows that $(I-R)^{-1}_{ij}$ can be interpreted as the sum of the weights of all the walks from $j$ to $i$. This implies the walk-sum interpretation advocated by \citet{Malioutov2006} for $J^{-1}=(I-R)^{-1}$, with $R=I-J$ and $\G$ the graphical model of the GMRF constructed by using Definition \ref{def:gmrf}. Most importantly, it follows that the calculation of $J^{-1}$ is susceptible to a particular resummation technique from graph theory based on the structure of sets of walks, called the path-sum.
The \textit{``most basic result of algebraic graph theory"} (as 
described in \cite{Flajolet2009}) states that the powers of the 
adjacency matrix $\A_\G$ of a graph $\G$ generate all the walks\footnote{Also called paths, i.e.~trajectories on the graph.} 
on this graph \cite{Biggs1993}. This result extends to any matrix $
\A$, by constructing a graph $\G=(\mathcal{V},\mathcal{E})$, with 
vertex set $\mathcal{V}$ and edge set $\mathcal{E}$, and such 
that $\A_{ij}\neq0\,\iff\,(j,i)\in\mathcal{E}$. Then, the weight of the 
edge from vertex $j$ to vertex $i$ on $\G$ is $\A_{ij}$ and $(\A^
\ell)_{ij}$ is the sum of the weights of all walks of length $\ell$ from 
$j$ to $i$ \cite{Flajolet2009}. The weight of a walk is simply the 
ordered product of the weight of the edges it traverses. Note how 
the indices of a matrix, e.g.~$\A_{ij}$, are written right-to-left and 
correspond to an edge, $(j,i)$, written left-to-right. This is because 
of unfortunate conventions. 

Now consider $\A$, the time-continuous motion generator of Definition \ref{DefTimeCont}. Since we have $(\I-\A)^{-1}=\sum_{\ell\geq0} \A^\ell$, it follows that $(\I-\A)^{-1}_{ij}$ can be interpreted as the sum of the weights of all the walks from $j$ to $i$ on a graph $\G$ associated with $\A$. Thus, the time-continuous propagator $\T$ and the time-ordered exponential of $\H$, both have an interpretation as a sum of weighted walks since, by Theorem \ref{MainRes}, $\T=(\I-\A)^{-1}\C$ and $\Theta(t'-t)\OE[\H](t',t)=\Phi^{-1}(\T)$. The time-ordered exponential is therefore susceptible to a particular resummation technique based on the structure of walk sets, called the method of path-sums.

In its most general form, the method of path-sums stems from a fundamental algebraic property of the ensemble of all walks on any weighted graph: any walk factorizes uniquely into products of prime walks, the \emph{simple paths} and \emph{simple cycles} of $\G$. A simple path, also known as self-avoiding walk, is a walk whose vertices are all distinct. A simple cycle, also known as self-avoiding polygon, is a walk whose endpoints are identical and intermediate vertices (from the second to the penultimate) are all distinct and different from the endpoints. The path-sum formulation for the sum of all walks weights on $\G$ is the representation of this sum that only involves the primes. 
Since there are only finitely many primes on any finite graph, the path-sum involves finitely many terms.
%We shall see that finite graphs sustain only finitely many primes, and thus the sum of walk weights has an exact representation involving only finitely many terms. 
For a full exposition of the algebraic structure of walk ensembles at the origin of the path-sum representation, we refer the reader to \cite{Giscard2012}. 
In \S\ref{PSOE} we give the path-sum formulation for the time-ordered exponential $\OE[\H](t',t)$. We shall see that it takes the form of a branched continued fraction of finite depth and breadth.

\subsection{Path-sum formulation of the time-ordered exponential}\label{PSOE}
We consider $\H(t)\in\mathbb{C}^{n\times n}[I^2]$ a time-dependent matrix. Let $\G=(\mathcal{V},\mathcal{E})$ be its associated graph: if there exists $t\in I$ with $\H_{\beta\alpha}(t)\neq0$, then $(\alpha,\beta)\in \mathcal{E}$. Let $\G\backslash\{\alpha_1,\cdots, \alpha_n\}$ denote the subgraph of $\G$ with the set of vertices $\{\alpha_1,\cdots, \alpha_n\}\subset \mathcal{V}$ and the edges incident to them removed from $\G$. Let $\Pi_{\G;\,\alpha\omega}$ and $\Gamma_{\G;\,\alpha}$ be the set of simple paths from $\alpha$ to $\omega$ on $\G$ and the set of simple cycles from $\alpha$ to itself on $\G$, respectively. If $\G$ has finitely many vertices and edges, these two sets are finite.

\vspace{1mm}\begin{theorem}\label{PSresult}
The time-ordered exponential of $\H(t)$ is given by
%\begin{subequations}
%\begin{align}
%&\mG_{\omega\alpha}=\sum_{p\in\Pi_{\G;\alpha\omega}}\mG_{\G\backslash\{\alpha,\nu_1,\cdots, \nu_{\ell-1}\};\,\omega}\cdots  \mG_{\G\backslash\{\alpha\};\,\nu_1} \mG_{\G;\,\alpha},\\
%\mG_{\G;\,\alpha}=\bigg[\I_\alpha-\sum_{\gamma\in \Gamma_{\G;\alpha}}\mG_{\G\backslash\{\alpha,\mu_1,\cdots, \mu_{\ell-2}\};\,\mu_{\ell-1}}\cdots  \mG_{\G\backslash\{\alpha\};\,\mu_1}\bigg]^{-1}.
%\end{align}
%\end{subequations}
%Equivalently,
\begin{subequations}
\begin{align}
&\hspace{-3mm}\OE_{\omega\alpha}[\H](t',t)=\sum_{p\in\Pi_{\G;\alpha\omega}}\!\!\int_{t}^{t'}\!\!\Big(\mG_{\G\backslash\{\alpha,\nu_1,\cdots, \nu_{\ell-1}\};\,\omega}\ast\H_{\omega\nu_{\ell-1}}\ast\cdots\label{OEaw}\\
&\hspace{38mm}\cdots  \ast\H_{\nu_2\nu_1}\ast\mG_{\G\backslash\{\alpha\};\,\nu_1}\ast\H_{\nu_1\alpha}\ast \mG_{\G;\,\alpha}\Big)\!(t',\tau)\,d\tau,\nonumber
\end{align}
where $p=(\alpha,\,\nu_1,\,\dotsc,\, \nu_{\ell-1},\,\omega)$ is a simple path of length $\ell$ from $\alpha$ to $\omega$ on $\G$ and $\matr{G}_{\G;\,\alpha}(t',t)$, called a Green's kernel, is given by
\begin{equation}
\mG_{\G;\,\alpha}=\bigg[1_\ast-\sum_{\gamma\in \Gamma_{\G;\alpha}}\!\!\mG_{\G\backslash\{\alpha,\mu_1,\cdots, \mu_{\ell'-2}\};\,\mu_{\ell'-1}}\ast\cdots\ast  \mG_{\G\backslash\{\alpha\};\,\mu_1}\ast\matr{H}_{\mu_1\alpha}\bigg]^{\ast-1},\label{GreenCycle}
\end{equation}
with $1_\ast:=\delta(t'-t)\matr{I}_n$ the time-dependent identity, $\gamma=(\alpha,\,\mu_1,\,\dotsc,\, \mu_{\ell'-1},\,\alpha)$ is a simple cycle of length $\ell'$ from $\alpha$ to itself on $\G$ and the time dependencies have been omitted for the sake of clarity.
\end{subequations}
\end{theorem}

\begin{remark} In the situation where $\alpha=\omega$, Eq.~(\ref{OEaw}) simplifies to
\begin{equation*}
\OE_{\G;\,\alpha\alpha}[\H](t',t)=\int_{t}^{t'}\!\mG_{\G;\,\alpha}(t',\tau)\,\, d\tau,
\end{equation*}
with $\mG_{\G;\,\alpha}$ the Green's kernel of Eq.~(\ref{GreenCycle}). It follows that this Green's kernel satisfies $\Theta(t'-t)\matr{G}_{\G;\,\alpha}(t',t)=\Phi^{-1}\big((\I-\A)^{-1}\big)_{\alpha\alpha}$.
\end{remark}

\begin{remark}
The Green's kernel $\mG_{\G;\,\alpha}$ is obtained recursively through Eq.~(\ref{GreenCycle}). Indeed, $\mG_{\G;\,\alpha}$ is expressed in terms of Green's kernels such as $\mG_{\G\backslash\{\alpha,\mu_1,\cdots, \mu_{\ell-2}\};\,\mu_{j-1}}$, which is in turn defined through Eq.~(\ref{GreenCycle}) but on the subgraph $\mathcal{G}\backslash{\{\alpha,\ldots,\mu_{j-1}}\}$ of $\mathcal{G}$. The recursion stops when vertex $\mu_j$ has no neighbour on this subgraph, in which case
\begin{align*}
  \mG_{\G\backslash\{\alpha,\mu_1,\cdots, \mu_{\ell-2}\};\,\mu_{j-1}} =  \begin{cases}
   \big[1_\ast-\H(t)_{\mu_j\mu_j}\big]^{\ast-1} & \text{if $\H(t)_{\mu_j\mu_j}\neq0$ for some $t\in I$,}\\
  1_\ast & \text{otherwise,}
  \end{cases}  
\end{align*}
The Green's kernel  $\mG_{\G;\,\alpha}$ is thus a branched continued fraction of $\ast$-inverses and which terminates at a finite depth. 
\end{remark}

\begin{proof}
We first show that any entry of the time-ordered exponential $\OE[\H](t',t)$ of $\H$ is a sum of weighted walks on the graph $\G$ associated with $\H$. We start with Eq.~(\ref{OEThm2}):
\begin{equation}
\Theta(t'-t)\OE[\H](t',t)_{\alpha \omega}=\int_{t}^{t'} \Phi^{-1}\big((\I-\A)^{-1}\big)_{\alpha \omega}(t',\tau)\,\,d\tau.\nonumber
\end{equation}
Since $\Phi$ is an homomorphism, $\Phi^{-1}\big(\A^{k}\big)(t',t)=(\matr{H}^{\ast k})(t',t)$ for $k\in\mathbb{N}$, and the Neumann series for $(\I - \A)^{-1}$ yields (we have shown in the proof of Theorem \ref{MainRes} that this series always exists)
\begin{align}\label{eq:fWalks2}
\Theta(t'-t)\OE[\H](t',t)_{\alpha\omega} &= \int_{t}^{t'}\!\!\!d\tau\,\sum^{\infty}_{k=0} (\matr{H}^{\ast k})_{\alpha\omega}(t',\tau),\nonumber\\
&=\int_{t}^{t'}\!\!d\tau\!\!\!\sum_{w\in W_{\G;\alpha\omega}} \!\!\!\!\ew(w)(t',\tau),
\end{align}
where the second equality follows by virtue of the correspondence between matrix powers and walks on the associated graph \cite{Flajolet2009, Giscard2013}. In particular, recall that the weight of a walk is the product of the weights of the edges it traverses, e.g. the walk $w=(\mu_0,\,\mu_1,\,\mu_2,\,\dotsc,\,\mu_n)$ has weight $\ew(w)=\H_{\mu_n\mu_{n-1}}\ast\cdots\ast\H_{\mu_2\mu_1}\ast\H_{\mu_1\mu_0}$.

Second, we use the result of \cite{Giscard2012}, which reduces a series of weighted walks, such as the one of Eq.~(\ref{eq:fWalks2}), into a sum of weighted simple paths and simple cycles. This result is reproduced here for the sake of completeness:
\begin{theorem}[Path-sum \cite{Giscard2012}]\label{PSformalresult}
Let $\G$ be a weighted graph with weight function $\ew(.)$. Assuming existence, 
the sum of the weights of all the walks from $\alpha$ to $\omega$ on $\G$, $
\sum_{w\in \W_{\mathcal{G};\alpha\omega}}\ew(w)$ admits the path-sum representation
\begin{subequations}\label{PathSumFormula}
\begin{align*}
 \sum_{w\in \W_{\mathcal{G};\alpha\omega}}\hspace{-2mm}\ew(w)=\sum_{\Pi_{\mathcal{G};\alpha\omega}}\ew(\omega)'_{\mathcal{G}\backslash\{\alpha,\ldots,\nu_\ell\}}\,\ew_{\omega\nu_\ell}\cdots \ew(\nu_2)'_{\mathcal{G}\backslash\{\alpha\}}\,\ew_{\nu_2\alpha}\, \ew(\alpha)'_\mathcal{G}\,,%\label{eqn:SumOverWalksToPaths2}
\shortintertext{where $\ew_{\mu\nu}$ is the weight associated with the edge $\nu\to\mu$ and $\ew(\alpha)'_{\mathcal{G}}$ is the sum of the weights of all walks from $\alpha$ to itself on $\G$. This quantity is recursively obtained as}
  \ew(\alpha)'_\mathcal{G}:=\Bigg[\matr{I} -\sum_{\Gamma_{\mathcal{G};\alpha}}\ew_{\alpha\mu_m}\,\ew(\mu_m)'_{\mathcal{G}\backslash\{\alpha,\ldots,\mu_{m-1}\}}\cdots \ew_{\mu_3\mu_2} \,
 \ew(\mu_2)'_{\mathcal{G}\backslash\{\alpha\}} \,\ew_{\mu_2\alpha}  \Bigg]^{-1},%\label{eqn:DressedVertexClosedForm2}
\end{align*}
\end{subequations}
with $\matr{I}$ the identity.
\end{theorem}

%The quantity $\ew(\alpha)'_{\mathcal{G}}$ is defined recursively through Eq.~(\ref{eqn:DressedVertexClosedForm2}). Indeed $\ew(\alpha)'_{\mathcal{G}}$ is expressed in terms of $\ew(\mu_j)'_{\mathcal{G}\backslash\{\alpha,\ldots,\mu_{j-1}\}}$ which is in turn defined through Eq.~(\ref{eqn:DressedVertexClosedForm2}) but on the subgraph $\mathcal{G}\backslash{\{\alpha,\ldots,\mu_{j-1}}\}$ of $\mathcal{G}$. The recursion stops when vertex $\mu_j$ has no neighbour on this subgraph, in which case $\ew(\mu_j)'_{\mathcal{G}\backslash\{\alpha,\ldots,\mu_{j-1}\}}=[\matr{I}-\ew_{\mu_j\mu_j}]^{-1}$ if the loop $(\mu_j\mu_j)$ exists and $\ew(\mu_j)'_{\mathcal{G}\backslash\{\alpha,\ldots,\mu_{j-1}\}}=\matr{I}$ otherwise. $\ew(\alpha)'_{\mathcal{G}}$ is thus a formal continued fraction which terminates at a finite depth. It can be seen as an effective vertex weight resulting from the dressing of vertex $\alpha$ by all the closed walks off $\alpha$ in $\mathcal{G}$.
%
%Now consider the partition
Theorem \ref{PSresult} follows from Theorem \ref{PSformalresult} on noting that an edge $\alpha\to\beta$ of the graph $\G$ corresponding to $\H$ has weight $\ew(\alpha\beta)=\H_{\beta\alpha}$ and that these weights multiply with the $\ast$-product. In particular, since $\ew(\alpha)'_\mathcal{G}$ is the sum of the weights of all walks from $\alpha$ to itself on $\G$, the correspondence between walks and matrix powers implies $\ew(\alpha)'_\mathcal{G}=\Phi^{-1}\big((\I-\A)^{-1}\big)_{\alpha\alpha}=\Theta(t'-t) \matr{G}_{\G;\,\alpha}(t',t)$.
\end{proof}

\section{Examples}
In this section we present three synthetic examples demonstrating the path-sum formulation of the time-ordered exponential, Theorem \ref{PSresult}.

\begin{figure}[t!]
\vspace{-25mm}
\begin{center}
\includegraphics[width=.3\textwidth]{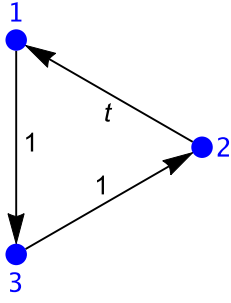}
\end{center}
\vspace{-5mm}
\caption{The oriented triangle $\mathcal{T}$, the graph corresponding to the time dependent matrix $\H(t)$ of Example \ref{TriangleExample}. The edge-weights are indicated in black, next to the edges.}\label{fig:T3}
\end{figure} 
\begin{example}\label{TriangleExample}
In this example, we consider the following system of differential equations on the interval $I=[0,T]$, $T>0$, for $t\in I$
\begin{equation*}
\H(t)v(t)=\dot{v}(t),\qquad
\H(t)=\begin{pmatrix}
0&t&0\\
0&0&1\\
1&0&0
\end{pmatrix}.
\end{equation*}
In spite of its apparent simplicity this system is not analytically solvable by the native DSolve function of \textit{Mathematica}. In particular, $\H(t)$ does not commute with itself at different times: $[\H(t'),\H(t)]=\H(t')\H(t)-\H(t)\H(t')\neq0$ for $t'\neq t$. The graph corresponding to $\H(t)$ is an \emph{oriented} triangle, which we denote $\mathcal{T}$, see Fig.~(\ref{fig:T3})

Let us first show in detail how the entry in the first row and the first column of the time-ordered exponential can be calculated. There is a single simple cycle from vertex 1 to itself on $\mathcal{T}$, namely $1\to3\to2\to1$. Furthermore $\mathcal{T}$ is acyclic as soon as a vertex is removed from it. Thus, following  Theorem \ref{PSresult}, the path-sum formulation for $\OE[\H](t,0)_{11}$ is:
\begin{align*}
&\OE[\H](t,0)_{11}=\int_{0}^t \mG_{\mathcal{T};11}(\tau) \,d\tau,\\[1ex]
& \mG_{\mathcal{T};11}(\tau)=\Big(1_\ast-\,\underbrace{\!\!\!\vphantom{\big[}
\overbrace{t}^{\text{Edge }1\ot2}\ast\overbrace{1}^{\text{Edge }2\ot 3}\ast \overbrace{1}^{\text{Edge }3\ot 1}}_{\text{Triangle}\CycleTri}\Big)^{\ast-1}(\tau,0).
%&f(t',t)=t\ast1\ast 1=\frac{1}{2} (t'-t)^2 t',\qquad{\text{cycle }1\to2\to3\to1.}\nonumber
\end{align*}
%The expression for $f(t',t)$ follows from noting that edges $(13)$, $(32)$ and $(21)$ have weights 1, 1 and $t$, respectively.
To find the required $\ast$-inverse, let $f(t',t):=\big(t\ast1\ast 1\big)(t',t)=(t'-t)^2 t'/2$. Then the $\ast$-powers of $f$ are given by
\begin{align}
f^{\ast n}(t,0)=\frac{4^{n+1}  \Gamma \left(n+\frac{7}{4}\right)}{\Gamma \left(\frac{3}{4}\right) \Gamma (4 n+4)}t^{3+4 n}\nonumber
\end{align}
for $ n > 0$ and $f^{\ast 0}(t,0)=\delta(t)$. The Neumann series $\sum_{n\geq 0}f^{\ast n}$ then gives
\begin{equation}
\big(1_\ast-t\ast1\ast1\big)^{\ast-1}(t,0)=\delta(t)+\frac{1}{2} t^3 \!\, _0F_2\left(;\frac{5}{4},\frac{3}{2};\frac{t^4}{64}\right)=\mG_{\mathcal{T};11}(t)\nonumber
\end{equation}
from which we find
\begin{equation*}
\OE[\H](t,0)_{11}=\int_{0}^t \mG_{\mathcal{T};11}(\tau)\,d\tau=\!\, _0F_2\left(;\frac{1}{4},\frac{1}{2};\frac{t^4}{64}\right).
\end{equation*}
Similarly
\begin{subequations}
\begin{align*}
\OE[\H](t,0)_{22}=\, _0F_2\left(;\frac{1}{2},\frac{3}{4};\frac{t^4}{64}\right),\qquad\OE[\H](t,0)_{33}=\, _0F_2\left(;\frac{1}{4},\frac{3}{4};\frac{t^4}{64}\right).
\end{align*}
\end{subequations}
Note how these entries differ from $\OE[\H](t,0)_{11}$, despite arising from circular permutations of the same triangle (cycles 1321, 2132 and 3213). This is because $t$ and $1$ do not $\ast$-commute, that is $t\ast1\neq 1\ast t$.

Since there is a single simple path between any two vertices of $\mathcal{T}$, path-sum formulation for the the off-diagonal elements is, for example,
\begin{subequations}
\begin{align*}
\OE[\H](t,0)_{12}&=\int_{0}^td\tau \,\Big(\overbrace{t}^{\text{Edge }2\ot1}\!\!\ast~~\mG_{\mathcal{T};22}\Big)(\tau),&\text{Path}\PathTwotoOne\\[-2ex]
&=\int_0^t\!d\tau \int_{0}^{\tau}\!d\tau'\,\tau\,\mG_{\mathcal{T};22}(\tau'),\\[2ex]
\OE[\H](t,0)_{13}&=\int_{0}^t\!d\tau\,\Big(\!\!\!\!\overbrace{t}^{\text{Edge }1\ot 2}\ast\overbrace{1}^{\text{Edge }2\ot 3}\!\!\ast~~\mG_{\mathcal{T};22}\Big)(\tau),&\text{Path}\PathThreeTwoOne\\[-2ex]
&=\int_0^t\!d\tau\!\int_0^{\tau} \!\!d\tau_1\!\int_{\tau_1}^{\tau}\!d\tau_2\,\tau\,\,\mG_{\mathcal{T};33}(\tau_1).
\end{align*}
Similarly,
\begin{align*}
 \OE[\H](t,0)_{21}&= \int_0^t\!d\tau\,\Big(1\ast1\ast\mG_{\mathcal{T};11}\Big)(\tau),&\text{Path}\PathOneThreeTwo\\
  \OE[\H](t,0)_{23}&=  \int_0^t \!d\tau\,\Big(1\ast\mG_{\mathcal{T}; 33}\Big)(\tau),&\text{Path}\PathThreeTwo\\
     \OE[\H](t,0)_{31}&=  \int_0^t\!d\tau\,\Big(1\ast\mG_{\mathcal{T}; 11}\Big)(\tau),&\text{Path}\PathOneThree\\
  \OE[\H](t,0)_{32}&=  \int_0^t\!d\tau\,\Big(1\ast t\ast\mG_{\mathcal{T};22}\Big)(\tau),&\text{Path}\PathTwoOneThree
%&=\frac{1}{2} t^3 \, _0F_2\left(;\frac{3}{4},\frac{3}{2};\frac{t^4}{64}\right)-\frac{1}{3} t^3 \,_0F_2\left(;\frac{1}{2},\frac{7}{4};\frac{t^4}{64}\right).
\end{align*}
\end{subequations}
%\begin{subequations}
%\begin{align*}
%\OE[\H](t,0)_{12}&=\int_0^t \!d\tau\,\tau\,\OE[\H](\tau)_{22},&\text{path }2\rightarrow 1\\
%%&=\frac{1}{2} t^2 \, _0F_2\left(;\frac{3}{4},\frac{3}{2};\frac{t^4}{64}\right),\\
%\OE[\H](t,0)_{13}&=\int_0^t \!d\tau_1\int_{\tau_1}^t\!d\tau_2\,\tau_2\,\OE[\H](\tau_1)_{33},&\text{path }3\rightarrow 2\rightarrow1\\
%%&=\frac{1}{2} t^3 \, _0F_2\left(;\frac{3}{4},\frac{5}{4};\frac{t^4}{64}\right)-\frac{1}{6} t^3 \,_0F_2\left(;\frac{1}{4},\frac{7}{4};\frac{t^4}{64}\right),\\
% \OE[\H](t,0)_{21}&=  \int_0^t \!d\tau_1\int_{\tau_1}^t\!d\tau_2\,\OE[\H](\tau_1)_{11},&\text{path }1\rightarrow 3\rightarrow2\\
%% &=t^2 \, _0F_2\left(;\frac{1}{2},\frac{5}{4};\frac{t^4}{64}\right)-\frac{1}{2} t^2 \, _0F_2\left(;\frac{1}{4},\frac{3}{2};\frac{t^4}{64}\right),\\
%  \OE[\H](t,0)_{23}&=  \int_0^t \!d\tau\,\,\OE[\H](\tau)_{33},&\text{path }3\rightarrow 2\\
%%&=t \, _0F_2\left(;\frac{3}{4},\frac{5}{4};\frac{t^4}{64}\right),\\
%     \OE[\H](t,0)_{31}&=  \int_0^t\!d\tau \,\,\OE[\H](\tau)_{11},&\text{path }1\rightarrow 3\\
%    % &=t \, _0F_2\left(;\frac{1}{2},\frac{5}{4};\frac{t^4}{64}\right),\\
%  \OE[\H](t,0)_{32}&=  \int_0^t\!d\tau_1\int_{\tau_1}^t \!d\tau_2\,\,\tau_1\,\OE[\H](\tau_1)_{22},&\text{path }2\rightarrow 1\rightarrow3
%%&=\frac{1}{2} t^3 \, _0F_2\left(;\frac{3}{4},\frac{3}{2};\frac{t^4}{64}\right)-\frac{1}{3} t^3 \,_0F_2\left(;\frac{1}{2},\frac{7}{4};\frac{t^4}{64}\right).
%\end{align*}
%\end{subequations}
Calculating these (straightforward) integrals, the time-ordered exponential of $\H$ is found to be
\begin{align}
&\OE[\H](t,0)=\\
&\hspace{-4mm}\begin{pmatrix}
\vphantom{\Big(}Q_{\frac{1}{4},\frac{1}{2}}(t)&\frac{1}{2} t^2 \,Q_{\frac{3}{4},\frac{3}{2}}(t)&\frac{1}{2} t^3\, Q_{\frac{3}{4},\frac{5}{4}}(t)-\frac{1}{6} t^3 \,Q_{\frac{1}{4},\frac{7}{4}}(t)\\
t^2 \,Q_{\frac{1}{2},\frac{5}{4}}(t)-\frac{1}{2} t^2\, Q_{\frac{1}{4},\frac{3}{2}}(t)&Q_{\frac{1}{2},\frac{3}{4}}(t)&t\, Q_{\frac{3}{4},\frac{5}{4}}(t)\\
t Q_{\frac{1}{2},\frac{5}{4}}(t)&
\vphantom{\Big(}\frac{1}{2} t^3 \,Q_{\frac{3}{4},\frac{3}{2}}(t)-\frac{1}{3} t^3\, Q_{\frac{1}{2},\frac{7}{4}}(t)&Q_{\frac{1}{4},\frac{3}{4}}(t)
\end{pmatrix}\nonumber
\end{align}
where $Q_{a,b}(t)=\!\,_0F_2\left(;a,b;\,t^4/64\right)$.
These expressions exactly match the numerical solutions obtained by \emph{Mathematica} and \emph{Matlab}. We can also verify $\H(t)\OE[\H](t,0)=\frac{d}{dt}\OE[\H](t,0)$ holds numerically within a machine epsilon.
\end{example}
~\\

\begin{figure}[t!]
\vspace{-20mm}
\begin{center}
\includegraphics[width=.35\textwidth]{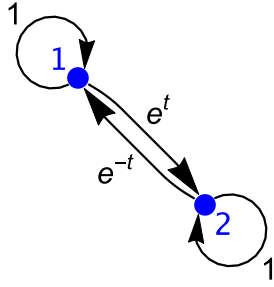}
\end{center}
\vspace{-5mm}
\caption{The complete graph on two vertices $\mathcal{K}_2$, the graph corresponding to the time dependent matrix $\H(t)$ of Example \ref{K2Example}. The edge-weights are indicated in black, next to the edges.}\label{fig:K2}
\end{figure}
\begin{example}\label{K2Example}
In this example, we consider the following system of differential equations on 
the interval $I=[0,T]$, $T>0$, for $t\in I$:
\begin{equation*}
\H(t)v(t)=\dot{v}(t),\qquad
\H(t)=\begin{pmatrix}
1&e^t\\
e^{-t}&1\\
\end{pmatrix}.
\end{equation*}
Again, this seemingly simple system is not analytically solvable by \textit{Mathematica}'s native DSolve function and $\H(t)$ does not commute with itself at different times. The graph $\G$ corresponding to $\H$ has the structure of the complete graph on 2 vertices with self-loops denoted $\mathcal{K}_2$, see Fig.~(\ref{fig:K2}).
According to Theorem \ref{PSresult}, the path-sum formulation for $\OE[\H]_{11}$ is 
\begin{align*}
&\OE[\H]_{11}(t,0)=\int_{0}^t \!d\tau\,\mG_{\G;11}(\tau),\\[1ex]
&\mG_{\G;11}=\bigg(1_\ast~~-\hspace{-2mm}\underbrace{1}_{\text{Self-loop}\LoopOneKTwo}\hspace{-2mm}-\underbrace{\overbrace{e^{-t}}^{\text{Edge }1\ot2}\ast~~\mG_{\G\backslash\{1\};22}~\,\ast\!\overbrace{e^t}^{\text{Edge }2\ot 1}}_{\text{Backtrack}\BackTrackOneTwo}\,\bigg)^{\ast-1}.
\end{align*}
Since the only cycle from $2$ to $2$ on $\G\backslash\{1\}$ is the self loop $2\to2$ with weight $1$, we have
\begin{equation*}
\mG_{\G\backslash\{1\};22}=\Big(1_\ast-\!\!\underbrace{1}_{\text{Self-loop}\LoopTwoKTwo}\Big)^{\ast-1}.
\end{equation*}
This quantity is easily obtained, $(1_\ast-1)^{\ast-1}(t',t)=\delta(t',t)+e^{(t'-t)}$. Then, the Neumann series for $\mG_{\G;11}$ gives
\begin{equation*}
\mG_{\G;11}(t)=\delta(t)+\frac{1}{10} e^{-\frac{1}{2} \left(\sqrt{5}-3\right) t} \left((5+\sqrt{5}) e^{\sqrt{5} t}+5-\sqrt{5}\right).
\end{equation*}
The entry $\OE[\H]_{11}(t,0)$ of the time-ordered exponential of $\H(t)$ then follows by integration. On using $\int_{0}^t \delta(\tau)d\tau=\Theta(0)=1$, by consistency with our definition of the Heaviside $\Theta$ function Eq.~(\ref{Heaviside}), we find 
\begin{align*}
&\OE[\H]_{11}(t,0)= e^{3t/2} \cosh
   \big(\sqrt{5}\,t/2\big)-\frac{e^{3 t/2}}{\sqrt{5}} \sinh\big(\sqrt{5}\,t/2\big).
\end{align*}
%\frac{1}{10} e^{-\frac{1}{2} \left(\sqrt{5}-3\right) t} \left(5+\sqrt{5}-\left(\sqrt{5}-5\right)
%   e^{\sqrt{5} t}\right).
Similarly, since the only simple cycle from 1 to itself on $\G\backslash\{2\}$ is the self-loop $1\to 1$, we have $\mG_{\G\backslash\{2\};11}=(1_\ast-1)^{\ast-1}$ and
\begin{align*}
\OE[\H]_{22}(t,0)&=\int_{0}^t \!d\tau\,\bigg(1_\ast~-\hspace{-5mm}\underbrace{1\vphantom{\big[}}_{\text{Self-loop}\LoopTwoTwoKTwo}\hspace{-5mm}-~\underbrace{\overbrace{e^{t}}^{\text{Edge }2\ot 1}\ast~\mG_{\G\backslash\{2\};11}~\ast \overbrace{e^{-t}}^{\text{Edge }1\ot 2}}_{\text{Backtrack}\BackTrackTwoOne}\,\bigg)^{\ast-1}\!\!(\tau),\\[1.5ex]
&=e^{t/2} \cosh\big(\sqrt{5}\,t/2\big)+\frac{e^{t/2}}{\sqrt{5}}\sinh\big(\sqrt{5}\,t/2\big).
\end{align*}
The path-sum formulation for the off-diagonal elements of the time-ordered exponential is easily obtained from the simple paths $2\ot1$ and $1\ot2$. We have
\vspace{-2mm}
\begin{align*}%\!\!\! \overbrace{e^{t}}^{\text{Edge }1\ot2}\!\!\!
\OE[\H]_{12}(t,0)&=\int_{0}^t\!d\tau\,\big(\mG_{\G\backslash\{2\};11}\,\ast e^t\ast\,\,\mG_{\G;22}\big)(\tau),&\text{Path}\PathTwoOne\\
%&=2\, e^{3 t/2} \sinh\big(\sqrt{5}\, t/2\big)/\sqrt{5}\\
\OE[\H]_{21}(t,0)&=\int_{0}^t\!d\tau\,\big(\mG_{\G\backslash\{1\};22}\ast e^{-t}\ast\mG_{\G;11}\big)(\tau),&\text{Path}\PathOneTwo
%&=2\, e^{t/2} \sinh\big(\sqrt{5}\, t/2\big)/\sqrt{5}.
\end{align*}
All together, the expression for the time-ordered exponential is
\begin{equation*}
\OE[\H](t)=e^{t/2}
\begin{pmatrix}
e^{t} a(t)-e^{t}b(t)&2\, e^{t} b(t)\\
   2\,b(t)& a(t)+b(t)
\end{pmatrix},
\end{equation*}
where $a(t):=\cosh\big(\sqrt{5}\,t/2\big)$ and $b(t):=\sinh\big(\sqrt{5}\,t/2\big)/\sqrt{5}$.
We verify analytically that this satisfies $\H(t)\OE[\H](t,0)=\frac{d}{dt}\OE[\H](t,0)$, as required.
\end{example}
~\\

\begin{example}[Time-dependent Dyson equation]
Let $\H(t)\in\mathbb{C}^{n\times n}[I^2]$ be a $2\times 2$ block matrix.
%This equation is in fact the simplest instance of the path-sum formula for the time-ordered exponential. Indeed, c
Let $\H_{\alpha\alpha}(t)$ and $\H_{\omega\omega}(t)$ designate the two diagonal blocks and $\G$ be the corresponding graph ($\mathcal{K}_2$), which has only two vertices, $\alpha$ and $\omega$. In this situation, $\H_{\alpha\alpha}(t)$ and $\H_{\omega\omega}(t)$ are the weights of the self-loops on vertices $\alpha$ and $\omega$, respectively. Similarly, $\H_{\omega\alpha}(t)$ and $\H_{\alpha\omega}(t)$ designate the off-diagonal blocks and are the weights of the edges $(\alpha\omega)$ and $(\omega\alpha)$, respectively.
Then Theorem \ref{PSresult} states that the Green's kernel at $\alpha$ is given by
\begin{align}\label{DysonInDisguise}
&\mG_{\mc{K}_2;\,\alpha}(t',t)=\Big[\I-\H_{\alpha\alpha}-\H_{\alpha\omega}\ast\mG_{\mc{K}_2\backslash\{\alpha\};\,\omega}\ast\H_{\omega\alpha}\Big]^{\ast-1},
\end{align}
and $\mG_{\mc{K}_2\backslash\{\alpha\};\,\omega}=(\I-\H_{\omega\omega}(t))^{\ast-1}$.
To show that this is the Dyson equation, we introduce 
\begin{subequations}
\begin{align*}
\mG&:=\mG_{\mc{K}_2;\,\alpha},\\
\mG_{0}&:=\big(\I-\H_{\alpha\alpha}(t)\big)^{\ast-1},\\
\Sigma&:=\H_{\alpha\omega}\ast\mG_{\mc{K}_2\backslash\{\alpha\};\,\omega}\ast\H_{\omega\alpha}\\
&\textcolor{white}{:}= \!\!\int_{t}^{t'}\!\!\!d\tau'\!\int_{t}^{\tau'}\!\!\!d\tau\,\H_{\alpha\omega}(t')\,\mG_{\mc{K}_2\backslash\{\alpha\};\,\omega}(\tau',\tau)\,\H_{\omega\alpha}(\tau).
\end{align*}
\end{subequations}
Eq.~(\ref{DysonInDisguise}) thus indicates that the Green's kernel $\mG$ obeys $\mG=(\mG_{0}^{\ast-1}-\Sigma)^{\ast-1}$, or equivalently,
\begin{equation*}
\mG=\mG_0+\mG \ast \Sigma\ast \mG_0.
\end{equation*}
This is the time-dependent Dyson equation \cite{Aoki2013}, which arises naturally from resummations of the Dyson series for the time-ordered exponential in the context of quantum many-body physics.
This equation appears for example when considering a physical entity (for example a particle, or an ensemble of sites in a solid) in contact with a large system. In this situation, the Hamiltonian driving the system + entity is naturally partitioned into four submatrices (blocks): $\H_{\alpha\alpha}(t)$, which drives the isolated entity, $\H_{\omega\omega}(t)$ which drives the rest of the system without the entity and $\H_{\omega\alpha}(t)$ and $\H_{\alpha\omega}(t)$, which represent the interactions between the system and the entity.

Hence we see that the Dyson equation stems from Theorem \ref{PSresult} on the complete graph on two vertices.
%The time-dependent Dyson equation is then stated as
%\begin{equation}\label{Dyson}
%\mG=\mG_0+\mG \ast \Sigma\ast \mG_0,
%\end{equation}
%with $\mG_0$ the time-ordered exponential of isolated entity, $\mG$ the time-ordered exponential of the entity in contact with the system and $\Sigma$, the self energy.
In general, Theorem \ref{PSresult} can be seen as extending the Dyson equation to an arbitrary number of systems/entities in contact with each other, and providing an \emph{explicit non-perturbative formula} for the self-energy that involves finitely many terms for finite systems.
\end{example}

\section{Decay properties}\label{decayprop}
In the last 15 years, a number of significant results have established exponentially decaying bounds for the magnitude of the entries of holomorphic functions of sparse matrices \cite{BenziNEW, Benzi1999, Benzi2007}. These results have given rise to a flurry of applications in linear algebra and physics as they underly efficient approximation techniques, see for example \cite{BenziNEW, Cramer2006, Shao2014}. As we shall see below, the techniques used to prove these results do not extend to the time-ordered exponential of time-dependent matrices which do not commute with themselves at different times.
In this section, we rely instead on Theorem \ref{MainRes} to establish the \textit{super-exponential decay} of the magnitude of the entries of the time-ordered exponential of sparse matrices.

\subsection{Super-exponential decay of the time-ordered exponential}
We begin by briefly recalling the existing results concerning holomorphic functions of sparse matrices. We follow the treatment of \cite{Benzi2007}.
Let $\M\in\mathbb{C}^{n\times n}$ be diagonalizable with spectral condition number $\kappa(\M)$. Let $f$ be an holomorphic function on an open subset $U$ of $\mathbb{C}$ containing the spectrum of $\M$. Let $\A_{\M}$ be the matrix defined by
\begin{equation*}
(\A_{\M})_{ij}:=\begin{cases}
1&\text{if }\M_{ij}\neq0,
\\0&\text{otherwise.}
\end{cases}
\end{equation*}
Let $\G_\M$ be the graph with adjacency matrix $\A_\M$, and let $d(i,j)$ be the length of the shortest path from vertex $i$ to vertex $j$ on $\G_\M$.
Then there exists positive constants $K$ and $0<\lambda<1$ such that \cite{Benzi2007}
\begin{equation}
|f(\matr{M})|_{j,i}\leq \kappa(\M) K \lambda^{d(i,j)}.\label{decaybound}
\end{equation}
This establishes the exponential decay of functions of $\M$ along its sparsity pattern $\G_\M$.
Ultimately, the proof of Eq.~(\ref{decaybound}) relies on the existence of a Cauchy integral representation for $f(\matr{M})$,
\begin{equation}\label{cauchy}
   f(\M) = \frac{1}{2\pi\,\mathrm{i}}\oint_{\Gamma} f(z)\,(z \matr{I}-\matr{M})^{-1}\,\mathrm{d}z,
\end{equation}
where $\rmi^2=-1$ and $\Gamma$ is a closed contour completely contained in $U$ that encloses the eigenvalues of $\matr{M}$. This representation is lacking in the case of the time-ordered exponential function of time-dependent matrices which do not commute with themselves at different times.

Since we cannot rely on a Cauchy representation to bound the entries of the time-ordered exponential of a sparse matrix, we turn instead to Theorem \ref{MainRes}.
The theorem establishes that the time continuous propagator $\T$ is proportional to  $(\I-\A)^{-1}$, which we have shown, can always be obtained from its Neumann series $\sum_{n\geq 0}\A^n$. Using the Taylor series remainder theorem then leads to the following bound for the entries of the time-ordered exponential of a sparse matrix:
\begin{proposition}\label{Propdecay}
Let $I\subset \mathbb{R}$ and $\H(t)\in\mathbb{C}^{n\times n}[I^2]$. Let $\matr{A}_\H$ be the adjacency matrix associated to $\H$ and $\G$ the corresponding graph. Let $d:=d(\alpha,\omega)$ be the length of the shortest path from vertex $\alpha$ to vertex $\omega$ on $\G$. Define $h:= \sup_{t\in I}\max_{\alpha,\beta}|\H_{\alpha\beta}(t)|$ and let $|\W_{\G;\alpha\omega;k}|$ be the number of walks of length $k$ from $\alpha$ to $\omega$ on $\G$
Then
\begin{subequations}
\begin{align}
|\OE[\H](t',t)_{\omega\alpha}|&\leq \sum_{k\geq d} \frac{h^k(t'-t)^k}{k!}|\W_{\G;\alpha\omega;k}|,\label{boundwalk}
\end{align}
with equality when $\matr{H}=h\, \matr{A}_{\H}$ is time independent. Let $\Delta$ be the maximum degree of any vertex of $\G$. If $\Delta$ is finite we also have the weaker bound
\begin{align}
|\OE[\H](t',t)_{\omega\alpha}|&\leq e^{\Delta h(t'-t)}\frac{\big(\Delta h(t'-t)\big)^{d}}{d!}.\label{genbound}
\end{align}
\end{subequations}
\end{proposition}

The bound of Eq.~(\ref{genbound}) demonstrates the super-exponential decay of the ordered-exponential function of any time-dependent sparse matrix, contingent on the assumption that $\Delta$ is finite. Furthermore, the result of the proposition is non-trivial only when the maximum distance $D:=\max_{\alpha,\omega}d(\alpha,\omega)$ between any two vertices $\alpha$ and $\omega$ on $\G$ is infinite. Otherwise, a super-exponentially decaying bound can always be found for any matrix, by choosing a large enough multiplying constant.

\begin{proof}
%We prove the decay results using the walk-sum representation of the time-ordered exponential.
Let $\alpha$ and $\omega$ be (possibly identical) vertices of the graph $\G$ associated with $\H$. Proceeding with Theorem 3.1 we have
\begin{equation}
\Theta(t'-t)\OE[\H](t',t)_{\omega\alpha}=\int_{t}^{t'} \Phi^{-1}\big((\I-\A)_{\omega\alpha}^{-1}\big)(t',\tau)\,\,d\tau.\nonumber
\end{equation}
Since $\Phi$ is an homomorphism, $\Phi^{-1}\big(\A^{k}\big)(t',t)=(\matr{H}^{\ast k})(t',t)$ for $k\in\mathbb{N}$, and the Neumann series for $(\I - \A)^{-1}$ yields
\begin{equation*}
\Theta(t'-t)\OE[\H](t',t)_{\omega\alpha} = \int_{t}^{t'}\!\!d\tau\,\sum^{\infty}_{k=0} \big(\matr{H}^{\ast k}\big)_{\omega\alpha}(t',\tau).
\end{equation*}
Writing the $\ast$-products explicitly (see Appendix \ref{StarPowersInv}) we find
\begin{align}
&\Theta(t'-t)\OE[\H](t',t)_{\omega\alpha} = \int_{t}^{t'}\!d\tau\, \sum_{k=0}^\infty \,\sum_{w\in \W_{\G;\alpha\omega;k}}\label{StepInt}\\
&\int_{\tau}^{t'}\! \!d\tau_{1} \int_{\tau_1}^{t'}\!\!d\tau_2\cdots \int_{\tau_{k-2}}^{t'}\!\!\!d\tau_{k-1}\,\,\matr{H}_{\nu_k \nu_{k-1}}(t',\tau_{k-1}) \cdots  \matr{H}_{\nu_1\nu_0}(\tau_2,\tau_1)\matr{H}_{\nu_1\nu_0}(\tau_1,t),\nonumber
\end{align}
with $\W_{\G;\alpha\omega;k}$ the ensemble of walks of length $k$ from $\alpha$ to $\omega$ on $\G$ and $w=(\nu_0\cdots \nu_k)\in \W_{\G;\alpha\omega;k}$.
Now let $h:= \sup_{t\in I}\max_{\alpha,\beta}| \H_{\alpha\beta}(t)|$ and $|\W_{\G;\alpha\omega;k}|$ be the number of walks of length $k$ from $\alpha$ to $\omega$ on $\G$. Taking the norm on both sides of Eq.~(\ref{StepInt}) we obtain, assuming $t'\geq t$,
\begin{align}
|\OE[\H](t',t)_{\omega\alpha}| &\leq \sum_{k\geq0}\sum_{w\in \W_{\G;\alpha\omega;k}} h^{k} \int_{t}^{t'}\!\!d\tau\int_{\tau}^{t'}\! \!d\tau_{1} \int_{\tau_1}^{t'}\!\!d\tau_2\cdots \int_{\tau_{k-2}}^{t'}\!\!\!d\tau_{k-1}\, 1,\nonumber\\
&=\sum_{k=0}^\infty h^{k} \frac{(t'-t)^{k}}{k!} |\W_{\G;\alpha\omega;k}|.\label{SumK}
\end{align}
Recall that the distance from $\alpha$ to $\omega$ on $\G$ is $d(\alpha,\omega)$. Since this distance is the length of the shortest path from $\alpha$ to $\omega$, there exists no walk from $\alpha$ to $\omega$ of length strictly shorter than $d(\alpha,\omega)$. Consequently, the set $\W_{\G;\alpha\omega;k}$ is empty when $k<d(\alpha,\omega)$ and thus $|\W_{\G;\alpha\omega;k}|=0$ for $k<d(\alpha,\omega)$. Therefore, the sum over $k$ in Eq.~(\ref{SumK}) starts at $k=d(\alpha,\omega)$. This proves Eq.~(\ref{boundwalk}).

If the number of walks of length $k$ from $\alpha$ to $\omega$ is unknown, we bound it by powers of the maximum degree $\Delta$ of any vertex of $\G$: $|\W_{\G;\alpha\omega;k}|\leq \Delta^k$ \cite{Chen2014}. Then, the Taylor series remainder theorem gives
\begin{equation*}
|\OE[\matr{H}](t',t)_{\alpha \omega}| \leq \sum_{k= d}^\infty \frac{h^{k} (t'-t)^{k}}{k!} \Delta^k\leq e^{\Delta h(t'-t)}\frac{\big(\Delta h(t'-t)\big)^{d}}{d!}\\
%&e^{\Delta h(t'-t)} \frac{\gamma (d, \Delta h(t'-t)}{(d-1)!}\\
%&= e^{\Delta h(t'-t)} \frac{(d-1)![1- e^{-(\Delta h(t'-t))} e_{d-1} (\Delta h(t'-t))]}{(d-1)!}
%\end{align}
%\end{subequations}
%\begin{equation}
%\|OE[H](t',t)_{\alpha \omega}\|
\end{equation*}
This proves the proposition.\end{proof}

\begin{remark}[Matrices partitions]\label{blockremark}
Proposition \ref{Propdecay} extends to arbitrary partitions of the matrix $\H$ into blocks on using the formalism for matrix partitions presented in \cite{Giscard2013}. In this case, the proposition is unchanged except that the result becomes a bound on $\|\OE[\H](t',t)_{\omega\alpha}\|$, the norm\footnote{Any sub-multiplicative norm.} of the $(\omega,\alpha)$-block of the time-ordered exponential. Furthermore, $h$ is now $h:= \sup_{t\in I}\max_{\alpha,\beta}\| \H_{\alpha\beta}(t)\|$, where $\| \H_{\alpha\beta}(t)\|$ is the norm of the block of $\H$ weighting the edge from $\beta$ to $\alpha$ on $\G$.
\end{remark}
%\textcolor{red}{
%\begin{remark}[A Better Bound On Time-Dependent Decay Properties]
%Although Eq.~(\ref{genbound}) gives a general upper bound, it can be further improved by expanding the exponential giving,
%\begin{equation}
%\|OE[H](t',t)_{\alpha \omega}\| \leq e^{\Delta h(t'-t)} \sum^{\infty}_{k=0} \frac{(\Delta h(t'-t))^{k}}{(d+k)k!}
%\end{equation}
%\end{remark}
%}
\subsection{Examples of super-exponential decay}
The general upper bound of Eq.~(\ref{genbound}) demonstrates \textit{super-exponential decay} of the time-ordered exponential of any finite time-dependent matrix $\H(t)$ as a function of the distance $d$ on the graph associated with $\H$, contingent on the assumption that $\Delta$ is finite.
%Surprisingly, the results of Theorem \ref{Propdecay} are also consistent with exponential decay in the situation where the largest eigenvalue of the graph is
% the maximum distance on the graph. More precisely, let $\{\G_n\}$ be a family of graphs with $\G_\infty$
%
%Let $\Delta_n$ be the largest eigenvalue of $\G_n$ and $d_n:=\max_{\alpha,\omega\in \G_n} d(\alpha,\omega)$. Since $d!\sim \sqrt{2 \pi d}(d/e)^d$ for $d\gg1$, we see that in the situation where $\Delta_n$ is comparable to $d_n$, the right hand side of Eq.~(\ref{genbound}) reads $e^{h(t'-t)}e^{d_n}/\sqrt{2\pi {d_n}}$, which fails to exhibit any decay with $d$.
%
%
% fails , the bound of  yields slower decay
Below we present explicit bounds for the ordered-exponential function of matrices with commonly encountered structures.
% In particular we
% some examples comparing general upper-bound of Eq.~(\ref{genbound}) with the walk-based bound Eq.~(\ref{boundwalk})
%
%
% and their specific upper-bound.
\begin{example}[Tridiagonal structure]
We consider the situation where $\H(t)$ has the structure of an infinite tridiagonal matrix at all times\footnote{As noted in Remark \ref{blockremark}, our results extend to the case where $\H(t)$ is block tridiagonal at all times.}. Then $\G$ is the infinite path-graph with a self-loop on each vertex. The number of walks of any length on this graph is $|\W_{\G;\alpha\omega;k}|= \sum _{j=0}^k \binom{k}{2j+d} \binom{2j+d}{j}$, with $d$ a shorthand notation for the distance $d(\alpha,\omega)$ from $\alpha$ to $\omega$. Eq.~(\ref{boundwalk}) then yields
\begin{equation}
|\OE[\H](t',t)_{\omega\alpha}|\leq e^{h(t'-t)} I_{d}\big(2h (t'-t)\big),\label{BesselTriDiag}
\end{equation}
with $I_n(z)$ the modified Bessel function of the first kind. This result is similar to that obtained by A. Iserles for the matrix exponential of tridiagonal matrices \cite{Iserles2000}. As expected, the above bound exhibits super-exponential decay with $d$ since $I_d(2x)\sim x^d/d!$ for $d\gg 1$.
As a comparison, the bound based on the maximum degree yields
$|\OE[\H](t',t)_{\alpha \omega}| \leq e^{2 h(t'-t)}\big(2 h(t'-t)\big)^{d}/d!$. For $d\gg 1$ this is larger than the bound of Eq.~(\ref{BesselTriDiag}) by a factor of $e^{2 h(t'-t)} 2^d$.
\end{example}

\begin{example}[Lattices]
We consider the situation where the graph $\G$ associated with $\H(t)$ has the structure of $\mathbb{Z}^\delta$, the infinite $\delta$-dimensional lattice with a self loop on each vertex ($\delta=1$ is the preceding case, $\delta=2$ is the square lattice with self loops etc.). Let $(0,0,\ldots,0)$ and $(a_1,\ldots,a_\delta)$ be the coordinates of $\omega$ and $\alpha$ on $\mathbb{Z}^\delta$, respectively, and let $d$ be the (Manhattan) distance from $\alpha$ to $\omega$. Then Eq.~(\ref{boundwalk}) gives
$$
|\OE[\H](t',t)_{\omega\alpha}|\leq e^{ h(t'-t)}\prod_{i=1}^{\delta}I_{a_i}\big(2h (t'-t)\big).
$$
This generalizes the super-exponential decay bound for tridiagonal matrices \cite{Iserles2000} to  $\mathbb{Z}^\delta$-structured matrices, $\delta\in\mathbb{N}\backslash\{0\}$. The bound is valid for both the exponential and time-ordered exponential functions. The weaker bound Eq.~(\ref{genbound}) is obtained with $\Delta=2\delta$ and
reads
\begin{subequations}
\begin{align}
|\OE[\H](t',t)_{\alpha\omega}|  \leq e^{2\delta h(t'-t)}(2\delta)^{d} \frac{h^d(t'-t)^{d}}{d!}.\nonumber
\end{align}
\end{subequations}
For $d\gg 1$, this is larger than the bound obtained using the number of walks by a factor $e^{(2\delta-1) h(t'-t)}(2\delta)^d$.
\end{example}

\begin{example}[Bethe lattices]
We consider the situation where the graph $\G$ associated with $\H(t)$ has the structure of the $N^\text{th}$ Bethe lattice (also known as an infinite Cayley tree). This is the infinite regular tree with degree $N\geq 2$. We denote this graph by $\mathcal{B}_N$.
The number of walks of length $2n+d$ on $\mathcal{B}_{N+1}$ is 
\begin{align*}
&|\W_{\mathcal{B}_{N+1};\,\alpha\omega;\,2n+d}|=\\
&\hspace{5mm}\sum _{k=d}^{n+d} \binom{2 n+d}{n+d-k} N^{n+d-k} \frac{2 k-d+1}{n+k+1}\leq(n+1) \binom{2 n+d}{n} N^{n} \frac{d+1}{n+d+1}.
\end{align*}
It follows by Eq.~(\ref{boundwalk}) the entries of the time-ordered exponential of $\H$ are bounded by
\begin{equation*}
|\OE[\H](t',t)_{\alpha\omega}|\leq \frac{e^{M}}{M}(d+1) N^{-\frac{d}{2}} \Big(I_{d+1}\big(2 M \big)+M\, I_{d+2}\big(2M\big)\Big).
%\leq \frac{e^{2\sqrt{N} h \left(t'-t\right)} \left(2\sqrt{N} h \left(t'-t\right)\right)^d}{d!}\nonumber
\end{equation*}
where $M:= h(t'-t)\sqrt{N}$. For $d\gg1$, $I_d(2x)\sim x^d/d!$ and the above bound exhibits super-exponential decay, as expected. 
%The largest eigenvalue for the graph of an infinite Bethe lattice with connectivity of $N+1$ is $\Delta = 2\sqrt{N}$, which gives the weaker bound $|\OE[\H](t',t)_{\alpha\omega}|\leq e^{2M} (2M)^d/d!$\,.
\end{example}

\subsection{Failure of super-exponential decay for infinite $\Delta$}
The second bound of Proposition \ref{Propdecay} for the entries of the time-ordered exponential of a sparse matrix $\H$, Eq.~(\ref{genbound}), requires the maximum degree $\Delta$ of any vertex of $\G$, the graph associated with $\H$, to be finite. In this section, we demonstrate a concrete example from quantum physics where $\Delta$ is infinite and Proposition \ref{Propdecay} yields an exponential rather than super-exponential decay for the entries of the time-ordered exponential of a time-dependent matrix.

\begin{example}[Exponential decay: time-dependent quantum spin systems]\label{hypercubex}
%Our last explicit bound concerns the time-ordered exponential of a matrix $\H(t)$ such that there exists a partition $\matr{H}(t)$, $\G$ is the $N$-hypercube $\mathcal{H}_N$ at all times.
A one-dimensional quantum spin system is a chain of spins: particles which can only be in two states, called up and down. We consider the situation where the spins have fixed positions, e.g. at regular intervals from one another along a straight line, and denote these positions $1$ to $N$. Neglecting contacts with the environment, the system is described by a Hamiltonian, an operator representing its energy. For a finite number $N$ of particles, we denote this Hamiltonian $\matr{H}_N(t)$, and take it to be 
\begin{equation}\label{HamilSpin}
\matr{H}_N(t)=\sum_{i=1}^N \big\{\sigma_x^i+J(t) \sigma_z^i \sigma_z^{i+1}\big\},
\end{equation}
where $$\sigma_x=\begin{pmatrix}0&1\\1&0\end{pmatrix}, \quad \sigma_z=\begin{pmatrix}1&0\\0&-1\end{pmatrix},$$ are the Pauli matrices, $J(t)$ is a real-valued time dependent function and the sum runs over the spins, with $i+1$ being the left (or right) neighbour of the $i$th spin. This Hamiltonian, known as the 1D antiferromagnetic Heisenberg model, has been extensively studied in physics and remains the subject of active research \cite{Bougourzi1996, Mikeska2004, Rabhi2006}.
The time-ordered exponential of this Hamiltonian is desirable since it appears for example in the time-dependent quasi-equilibrium partition function $\mathsf{Z}(t',t):=\OE[-\H_N(t)/kT](t',t)$ with $k$ Boltzmann constant and $T$ the temperature, which dictates thermodynamic properties of the system \cite{Hummer2001,Jarzynski1997}.

Structurally, we verify that at all times, $\H_N(t)$ is the $N$th-hypercube $\mathcal{H}_N$, with one self loop on each vertex. The hypercube $\mathcal{H}_N$ has the property that both its maximum degree and the maximum distance $D_N$ between to vertices are equal to $N$, $\Delta_N=N$ ($\Delta_N=N+1$ when self-loops are present) and $D_N:=\max_{\alpha,\omega}d(\alpha,\omega)=N$. 
%
%The family of finite time-dependent matrices $\{\H_N(t)\}_{N\in\mathbb{N}}$ such that $\H_\infty:=\lim_{N\to\infty}\H_N$ exists at all times. Let $\Delta_N$ be the maximum degree of the graph $\G_N$ associated with $\H_N$ and let $D_N$ be the maximum distance between any two vertices of $\G_N$.  In the situation where $\lim_{N\to\infty }\Delta_N/D_N=1$, we show below that the time-ordered exponential of $\H_\infty$ exhibits only exponential -- rather than super-exponential -- decay. Surprisingly, we will see that this remains in agreement with Proposition \ref{Propdecay}. This situation is commonly encountered in the physics of quantum spin systems.
In particular remark that $\lim_{N\to\infty}\Delta_N/D_N=1$.

Using Proposition \ref{Propdecay}, we bound the entries of $\mathsf{Z}(t',t)$ for $N$ finite. We then study the behavior of the bound under the limit $N\to\infty$.
The number of walks of length $2n+d$ between any two vertices at distance $d$ from each other on the $N$-hypercube is given by
\begin{equation*}
|\W_{\mc{H}_N;\,\alpha\omega;\,2n+d}|=\frac{2}{2^N}\sum_{i=0}^{\lfloor N/2\rfloor}(2 i+p_N)^{2n+d}\sum_{j=0}^{\lfloor N/2\rfloor}\binom{N-d}{\lfloor \frac{N}{2}\rfloor-i-j}\binom{d}{j}(-1)^j.\label{walkshyper}
\end{equation*}
with $\lfloor .\rfloor$ the floor function and $p_N = N\!\!\mod 2$. Then, following Eq.~(\ref{boundwalk}), the bound for the entries the time-dependent quasi-equilibrium partition function $\mathsf{Z}(t',t)$ is
\begin{equation}\label{boundhyp}
|\mathsf{Z}(t',t)_{\omega\alpha}|\leq e^{\beta h (t'-t)}\sinh ^d\left(\beta h \left(t'-t\right)\right) \cosh ^{N-d}\left(\beta h \left(t'-t\right)\right),
\end{equation}
where $\beta:=1/kT$. The bound only exhibits exponential rather than super-exponential decay as a function of $d\leq D_N$. Indeed, let $x:=\beta h(t'-t)$, $\rho(x):=\sinh(x)/\cosh(x)<1$. The bound Eq.~(\ref{boundhyp}) is seen to decay proportionally to 
\begin{equation*}
\sinh(x)^d\cosh(x)^{N-d}\propto \cosh(x)^N \rho(x)^d.
\end{equation*}
In particular, when $x<1$, the bound decays as $\cosh(x)^{N} \rho(x)^d=x^r+O(x^2)$ with $r:=d/D_N$. In both case, the decay is only exponential.
This remains true in the infinite system limit $N\to\infty$. Since Eq.~(\ref{boundhyp}) is saturated in the (physically trivial) situation where $J(t)=0$, we are forced to conclude that in general, the time-ordered exponential function does not exhibit super-exponential decay when $\Delta$ is infinite.

Remark that the result of Eq.~(\ref{boundhyp}) is compatible with Eq.~(\ref{genbound}), since Eq.~(\ref{genbound}) gives
\begin{equation*}
|\mathsf{Z}(t',t)_{\omega\alpha}|\leq e^{(N+1)\beta h (t'-t)} \big((N+1)\, \beta h (t'-t)\big)^d/d!\,.
\end{equation*}
For any $N$ finite, $d\leq D_N$, and the above is larger than the bound of Eq.~(\ref{boundhyp}).
\end{example}

\section{Discussion: analytical and numerical calculations}\label{discuss}
The path-sum formulation of the time-ordered exponential presented in this article provides a new approach to obtaining exact solutions to systems of differential equations with variable coefficients. Since the full solution requires evaluating multiple $\ast$-inverses, each of which corresponds to solving a Volterra equation of the second kind or evaluating a Neumann series, a completely analytical solution appears intractable in many cases. However, well-established numerical methods for solving Volterra equations of the second kind exist \cite{Linz1985, Brunner1986}, and so the path-sum formulation presented here is expected to be accessible to numerical algorithms.

The interpretation of the time-ordered exponential as a sum of walk weights which underlies Theorem \ref{PSresult} opens the door to an alternative formulation of the problem, which may be even more suitable for numerical evaluations. The central idea is to exploit recent results regarding the algebraic structure of the set of walks on an arbitrary digraph \cite{Thwaite2014} to identify certain infinite geometric series of terms appearing in the sum of walk weights. These geometric series can then be exactly resummed, thereby reducing the sum of all walk weights to a sum over the weights of a certain (possibly infinite) subset of `irreducible' walks. Each term in this sum is modified so as to exactly include the contributions of the infinite families of resummed terms. The exact form of the irreducible terms remaining in the sum depend on the structure of the resummed terms: choosing a different family of terms produces a different series. In this perspective, the path-sum formulation presented in this work corresponds to the extreme situation where all possible resummations have been made.  

Instead, one may choose an intermediate situation where the series of walks is only partially resummed. This will leave infinitely many terms to sum over, however an advantage of this strategy is that these terms only contains `easy' $\ast$-inverses, more precisely $\ast$-inverses of polynomials. These intermediate formulations therefore form a promising starting point for developing numerical approaches to the time-ordered exponential. We will present examples of numerical calculations of the time-ordered exponential based on these graphical resummation techniques in a future work.

\section*{Acknowledgement}
\noindent P.-L.~Giscard and D.~Jaksch acknowledge funding from EPSRC Grant EP/K038311/1. D.~Jaksch also received funding from the ERC under the European Unions Seventh Framework Programme (FP7/2007-2013)/ERC Grant Agreement no. 319286 Q-MAC. K.~Lui was funded by the Bechtel Fund Summer Internship Award. S.~J.~Thwaite acknowledges funding from the Alexander von Humboldt Foundation. 

 \appendix
\section{Powers and inverses associated with the $\ast$-product}\label{StarPowersInv}
\subsection{Time-dependent matrices: $\ast$-powers}
The $\ast$-product naturally induces the notion of $\ast$-powers. Since $\Phi$ is an homomorphism, $\Phi(\matr{m}^{\ast n})=\matr{M}^n$ and the $\ast$-powers are given by the formula
\begin{align}\label{StarPowern}
&\matr{m}^{\ast n}(t',t)=\int_{t}^{t'} \!d\tau_{1} \int_{\tau_1}^{t'}\!d\tau_2\cdots \int_{\tau_{n-2}}^{t'}\!\!d\tau_{n-1}\, \,\matr{m}(t',\tau_{n-1})\cdots \matr{m}(\tau_2,\tau_1)\matr{m}(\tau_1,t),
\end{align}
for $[t',t]\subseteq I$. In the situation where $\matr{m}(t',t)=\matr{m}(t'-t)$, these integrals become nested convolutions and can be tackled in the Fourier or Laplace domains.

\subsection{Time-dependent matrices: $\ast$-inverses}\label{AstInverse}
The $\ast$-inverse $\matr{m}^{\ast-1}$ of a time-dependent matrix $\matr{m}$  is the time-dependent matrix solution of $\matr{m}\ast\matr{m}^{\ast-1}=1_\ast$, where $1_\ast:=\delta(t'-t)\,\I_n$ is the time-dependent identity.

Since all the inverses appearing in our work are of the form $(1_\ast-\matr{m})^{\ast-1}$, we give the corresponding defining equation:
\begin{equation}\label{InverseEquation}
\matr{m}\ast(1_\ast-\matr{m})^{\ast-1}=(1_\ast-\matr{m})^{\ast-1}-1_\ast.
\end{equation}
%
%
%from the Neumann series. Supposing convergence of the series
%
%we readily define the $\ast$-inverse $\matr{m}^{\ast-1}$ of a time-dependent matrix $\matr{m}$ as fulfilling %The $\ast$-inverse fulfills
%\begin{equation}
%\matr{f}\ast\matr{f}^{\ast-1}=1_\ast.
%%\end{equation}with $1_\ast$ the identity matrix
%\begin{equation}
%f^{\ast-1}=\lim_{N\to\infty}\sum_{n=0}^N (1_{\ast}-f)^{\ast n}
%\end{equation}
%or equivalently
Such an inverse can be calculated from its Neumann series
\begin{equation}
(1_\ast-\matr{m})^{\ast-1}=\sum_{n}\matr{m}^{\ast n},
\end{equation}
as this series always converges. Indeed we have required that for any matrix $\matr{m}\in\mathbb{C}^{n\times n}[I^2]$, the quantity $\sup_{t_1,t_2\in I}\|\matr{m}(t_2,t_1)\|\leq b$ be bounded. Therefore, by Eq.~(\ref{StarPowern}), $\sup_{t_1,t_2\in I}\|\matr{m}^{\ast n}(t_2,t_1)\|\leq b^n/n!$ and $\sum_{n}\matr{m}^{\ast n}$ is absolutely convergent and $(1_\ast-\matr{m})^{\ast-1}$ exists.

Alternatively the inverse can be obtained on noting from Eq.~(\ref{InverseEquation}) that $\ast$-inverses are solutions to Volterra equations. More precisely, let $\matr{h}(t',t)$ such that $1_\ast(t',t)+\matr{h}(t',t)=(1_\ast-\matr{m})^{\ast-1}(t',t)$. Then
% The above is then
%$$f\ast(1_\ast-f)^{\ast-1}=f(t',t)+\int_{t}^{t'}\!d\tau\,\,f(t',\tau)h(\tau,t)=h(t',t)$$
%Hence, we are looking for
$\matr{h}$ solves a Volterra equation of the second kind:
\begin{equation}
\matr{h}(t',t)=\matr{m}(t',t)+\int_{t}^{t'}\!d\tau\,\,\matr{m}(t',\tau)\matr{h}(\tau,t).\label{VolterraInverse}
\end{equation}
In other terms, $(1_\ast-\matr{m})^{\ast-1}-1_\ast$ is the (matrix) Green's function of a Volterra equation with integral kernel $\matr{m}$.
Finally, note that since $\Phi$ is an isomorphism $\Phi(\matr{m}^{\ast-1})=\matr{M}^{-1}$ and the $\ast$-inverse corresponds to an inverse along the continuous time index.

\bibliographystyle{elsarticle-num}
\bibliography{Ordered_Exp}

\begin{thebibliography}{10}
\expandafter\ifx\csname url\endcsname\relax
  \def\url#1{\texttt{#1}}\fi
\expandafter\ifx\csname urlprefix\endcsname\relax\def\urlprefix{URL }\fi
\expandafter\ifx\csname href\endcsname\relax
  \def\href#1#2{#2} \def\path#1{#1}\fi

\bibitem{Golub}
G.~H. Golub, C.~F.~V. Loan, {Matrix Computations}, 4th Edition, Johns Hopkins
  University Press, 2012.

\bibitem{Higham2008}
N.~J. Higham, Functions of Matrices: Theory and Computation, 1st Edition,
  Society for Industrial and Applied Mathematics, Philadelphia, 2008.

\bibitem{Lam1998}
C.~S. Lam, {Decomposition of time-ordered products and path-ordered
  exponentials}, J. Math. Phys. 39 (1998) 5543--5558.

\bibitem{Blanes2009}
S.~Blanes, F.~Casas, J.~A. Oteo, J.~Ros, {The Magnus expansion and some of its
  applications}, Physics Reports 470 (2009) 151--238.

\bibitem{Georges1996}
A.~Georges, G.~Kotliar, W.~Krauth, M.~J. Rozenberg, {Dynamical mean-field
  theory of strongly correlated fermion systems and the limit of infinite
  dimensions}, Rev. Mod. Phys. 68 (1996) 13--125.

\bibitem{Aoki2013}
H.~Aoki, N.~Tsuji, M.~Eckstein, M.~Kollar, T.~Oka, P.~Werner, {Nonequilibrium
  dynamical mean-field theory and its applications}, Rev. Mod. Phys. 86 (2014)
  779.

\bibitem{Flajolet2009}
P.~Flajolet, R.~Sedgewick, Analytic Combinatorics, 1st Edition, Cambridge
  University Press, Cambridge, 2009.

\bibitem{Biggs1993}
N.~Biggs, Algebraic Graph Theory, 2nd Edition, Cambridge University Press,
  Cambridge, 1993.

\bibitem{Giscard2012}
P.-L. Giscard, S.~J. Thwaite, D.~Jaksch, {Continued fractions and unique
  factorization on digraphs}, arXiv:1202.5523.

\bibitem{Giscard2013}
P.-L. Giscard, S.~J. Thwaite, D.~Jaksch, Evaluating matrix functions by
  resummations on graphs: the method of path-sums, SIAM. J. Matrix Anal. \&
  Appl. 34 (2013) 445--469.

\bibitem{BenziNEW}
M.~Benzi, P.~Boito, N.~Razouk, Decay properties of spectral projectors with
  applications to electronic structure, SIAM Review 55 (2013) 3--64.

\bibitem{Benzi1999}
M.~Benzi, G.~H. Golub, Bounds for the entries of matrix functions with
  application to preconditioning, BIT 39 (1999) 417--438.

\bibitem{Benzi2007}
M.~Benzi, N.~Razouk, Decay bounds and \textit{{O}}$(n)$ algorithms for
  approximating functions of sparse matrices, Electron. T. Numer. Ana. 28
  (2007) 16--39.

\bibitem{Cramer2006}
M.~Cramer, J.~Eisert, Correlations, spectral gap and entanglement in harmonic
  quantum systems on generic lattices, New J. Phys. 8 (2006) 71.

\bibitem{Shao2014}
M.~Shao, On the finite section method for computing exponentials of
  doubly-infinite skew-hermitian matrices, Linear Algebra and its Applications
  451 (2014) 65--96.

\bibitem{Chen2014}
X.~Chen, J.~Qian, Bounds on the number of closed walks in a graph and its
  applications, Journal of Inequalities and Applications 2014 (2014) 199.

\bibitem{Iserles2000}
A.~Iserles, How large is the exponential of a banded matrix?, New Zealand J.
  Math. 29 (2000) 177--192.

\bibitem{Bougourzi1996}
A.~H. Bougourzi, M.~Couture, M.~Kacir, Exact two-spinon dynamical correlation
  function of the one-dimensional heisenberg model, Phys. Rev. B 54 (1996) R
  12669.

\bibitem{Mikeska2004}
H.-J. Mikeska, A.~K. Kolezhuk, Lecture notes in physics, Berlin: Springer 645
  (2004) 1--83.

\bibitem{Rabhi2006}
A.~Rabhi, P.~Schuck, J.~da~Provid\^{e}ncia, Random phase approximation for the
  1d anti-ferromagnetic heisenberg model, J. Phys.: Condens. Matter 18 (2006)
  10249.

\bibitem{Hummer2001}
G.~Hummer, A.~Szabo, Free energy reconstruction from nonequilibrium
  single-molecule pulling experiments, Proc. Natl. Acad. Sci. USA 98 (2001)
  3658.

\bibitem{Jarzynski1997}
C.~Jarzynski, Nonequilibrium equality for free energy differences, Phys. Rev.
  Lett. 78 (1997) 2690.

\bibitem{Linz1985}
P.~Linz, Analytical and numerical methods for Volterra equations, 1st Edition,
  Society for Industrial and Applied Mathematics (SIAM), Philadelphia, 1985.

\bibitem{Brunner1986}
H.~Brunner, P.~J. van~der Houwen, The numerical solution of Volterra equations,
  1st Edition, Oxford : North-Holland, Amsterdam, 1986.

\bibitem{Thwaite2014}
S.~J. Thwaite, {A family of partitions of the set of walks on a directed
  graph}, arXiv:1409.3555\href {http://arxiv.org/abs/1409.3555}
  {\path{arXiv:1409.3555}}.

\end{thebibliography}

\end{document}